\documentclass[a4paper, 11pt]{article}
\usepackage{geometry}
\usepackage[english]{babel}
\usepackage{STYLEFILE}
\usepackage{STYLEFILE_THEOREM_EN}

\title{Quantum Algorithms based on the Block-Encoding Framework for Matrix Functions by Contour Integrals}
\author[1,4]{Souichi Takahira\thanks{Email:takahira@qc.ee.es.osaka-u.ac.jp}} 
\author[2]{Asuka Ohashi}
\author[3]{Tomohiro Sogabe}
\author[4]{Tsuyoshi Sasaki Usuda}

\affil[1]{
    Graduate School of Engineering Science, Osaka University, \authorcr
    \textit{1-3 Machikaneyama, Toyonaka, Osaka, 560-8531, Japan}
}
\affil[2]{
    Department of General Education, 
    National Institute of Technology, Kagawa College, \authorcr
   \textit{551 Kohda, Takuma-cho, Mitoyo, Kagawa, 769-1192, Japan}
}
\affil[3]{
    Graduate School of Engineering, Nagoya University, 
    \textit{Furo-cho, Chikusa, Nagoya, 464-8603, Japan}
}
\affil[4]{
    Graduate School of Information Science and Technology, Aichi Prefectural University, \authorcr
    \textit{1522-3 Ibaragabasama, Nagakute, Aichi, 480-1198, Japan}
}

\date{\today}

\begin{document}

\maketitle

\begin{abstract}
      The matrix functions can be defined by Cauchy's integral formula 
  and can be approximated by the linear combination of 
  inverses of shifted matrices using a quadrature formula. 
  In this paper, we show a concrete construction of a framework to implement the linear combination of the inverses on quantum computers
  and propose a quantum algorithm for matrix functions based on the framework. 
  Compared with the previous study 
  [S. Takahira, A. Ohashi, T. Sogabe, and T.S. Usuda, Quant. Inf. Comput., 
  \textbf{20}, 1\&2, 14--36, (Feb. 2020)] that proposed a quantum algorithm 
  to compute a quantum state for the matrix function based on the circular contour centered at the origin, 
  the quantum algorithm in the present paper can be applied to a more general contour. 
  Moreover, the algorithm is described by the block-encoding framework.  
  Similarly to the previous study, the algorithm can be applied even if the input matrix is not a Hermitian or normal matrix. 
    
\end{abstract}

\section{Introduction}

\subsection{Background}
In 2009, Harrow, Hassidim and Lloyd proposed a quantum algorithm for linear systems that outputs a quantum state corresponding to a normalized solution vector \cite{HHL2009}.  
Many quantum algorithms of the same kind \cite{CKS2017, SBJ2019, AGGW2020}, which solve linear algebra tasks, 
have been developed after the breakthrough by Harrow et al.
These algorithms are applied to many quantum algorithms 
(for example, differential equations \cite{Berry14, BCOW2017, CLO2020}, 
quantum machine learning \cite{WBL2012, RML2014}, semidefinite programs \cite{AGGW2020}) as a subroutine. 
In particular, quantum algorithms for matrix functions are important 
because  computing the matrix function is one of the important tasks in linear algebra \cite{Higham}. 

Matrix functions have different representations. 
Among several equivalent definitions of matrix functions, 
we adopt a definition by a contour integral. 
Let $A$ be a square matrix. 
Suppose that $\Gamma$ is a closed contour enclosing all the eigenvalues of $A$ and 
$f$ is an analytic function on and inside $\Gamma$. 
Then, the matrix function $f(A)$ can be defined as follows: 
\begin{align}
  f(A) = \frac{1}{2\pi\im}\int_\Gamma f(z)(zI - A)^{-1} \dz, 
  \label{eq:deffA}
\end{align}
where $\im = \sqrt{-1}$ is the imaginary unit. 
In the classical computing, many methods for matrix functions based on different definitions have been developed. 
For the details of these definitions, see, e.g., \cite{Rine1955, Higham}. 
In \cite{TOSU2020}, we proposed a quantum algorithm for obtaining a state 
$\ket{f} = f(A)\ket{\psi}/\norm{f(A)\ket{\psi}}$ when the contour is a circle with centered at the origin. 
However, the quantum algorithm proposed in \cite{TOSU2020}
cannot be applied to a function (e.g., matrix logarithm) that has a singular point at the origin. 
As such matrix functions appear in practical problems \cite{logdet}, 
it is important to construct a quantum algorithm for them. 

In this paper, we show a quantum algorithm for computing the state $\ket{f} = f(A)\ket{\psi}/\norm{f(A)\ket{\psi}}$ 
even if $\Gamma$ is not a circle centered at the origin. 
Specifically, we consider quantum algorithms for computing a linear combination of inverse matrices. 
This is because the problem of computing a matrix of the form \eqref{eq:deffA} 
is essentially the problem of computing the linear combination of 
inverses of shifted matrices arising from a quadrature formula \cite{Higham}, 
e.g., trapezoidal rule, Gauss-Laguerre quadrature, and the double exponential formula \cite{TM74}.
Using the algorithm for the linear combination of inverses as a main component, 
we construct two quantum algorithms for matrix functions.
The first one is the case in which the contour is a circle and the quadrature method is the trapezoidal rule. 
The second one is a case in which the contour is not a circle but a general contour, leading to a framework for computing the approximations of matrix functions obtained by various quadrature formulas.

Unlike the previous study of \cite{TOSU2020}, 
we describe quantum algorithms in terms of block-encoding \cite{SGJ19, GSLW2018, GSLW2019}, 
which is a framework for implementing matrix arithmetic on quantum computers.  
Using this framework, we can concisely describe the circuit of the quantum algorithms to solve linear algebra tasks.
Moreover, we obtain a complexity, the number of ancilla qubits of the quantum algorithm, and an error occurred in the quantum algorithm. 
In addition, it is easy to employ the algorithm as a subroutine. 

The block-encoding of the linear combination of inverses consists of a block-encoding of the inverse of a block-diagonal matrix, 
which is also considered as in \cite{TOSU2020}, and a circuit similar to the well-known linear combinations of unitaries (LCU) approach  \cite{CW2012,BCCKS2015}. 
Such the construction is also considered in the study \cite[Section 5.1]{TAWL2020}. 
The study \cite[Section 5.1]{TAWL2020} discusses mainly the matrix exponential represented by the contour integral 
and approximated by Gauss-Laguerre quadrature. Compared with the study, 
we discuss a general matrix function and methods for the approximations by various quadratures in detail. 
In the preliminaries, we show constructions of block-encodings for basic operations. 
We use these constructions to obtain a block-encodings of the block-diagonal matrix and its inverse. 

We now state a relation between a matrix function and its trace. 
If the block-encoding of the matrix is constructed, 
then the trace of the block-encoded matrix can be encoded in the amplitude of a certain state using 
a circuit for generating the maximum entangled state and the Hadamard test \cite{AGGW2020}. 
By estimating the amplitude using techniques such as amplitude estimation \cite{BHMT2002}, we obtain the trace of the block-encoded matrix. 
In other words, if we construct a block-encoding of $f(A)$, then we can estimate $\Tr f(A)$. 
The trace of the matrix function relate to various quantities, 
for example, there is a relation between a log-determinant of matrices and a trace of the matrix logarithm: $\log \det A = \Tr \log A$. 
Therefore, our method can be used to estimate various quantities by combing the quantum trace estimation.

\subsection{Related work}

In this section, we  review some related studies on quantum algorithms for (general) matrix functions by aspects of the principle. 
A quantum algorithm for a polynomial of unitary matrices was proposed in \cite{KR2003}. 
For the case of non-unitary matrices, Harrow, Hassidim, and Lloyd proposed a quantum algorithm for implementing functions of 
sparse Hermitian matrices \cite{HHL2009} on quantum computers. 
The key idea of \cite{HHL2009} is the fact that $f(A)$ is approximated as 
$f(A)  = U f(\Lambda) U^\dag \approx U f(\tilde{\Lambda}) U^\dag$, where 
$A =  U^\dag \Lambda U$ is the eigendecomposition of $A$ and 
$\tilde{\Lambda}$ denotes a diagonal matrix in which the diagonal elements are approximations of the eigenvalues.
The eigenvalues estimation is performed by quantum phase estimation (QPE) with Hamiltonian simulation on the Hermitian matrix $A$ 
and the map $\lambda \mapsto f(\lambda)$ is performed by a rotation gate based on the estimated eigenvalues. 
When $f$ is a function close to $1/x$, this method outputs a state corresponding to the solution of the linear system. 
Because the method based on eigenvalue estimation uses QPE, 
the overall cost is $\poly(1/\epsilon)$, where $\epsilon$ is the desired accuracy for the state $\ket{f}$. 

An approach based on function approximation was also considered \cite{CKS2017}. 
In this approach, the matrix function is approximated as $f(A)= U f(\Lambda) U^\dag  \approx U \tilde{f}(\Lambda) U^\dag = \tilde{f}(A)$, 
where $\tilde{f}$ is an approximation of the original function $f$. 
In particular, approaches using Fourier series \cite{AGGW2020} or Chebyshev polynomials (of the first kind) \cite{SBJ2019} were considered. 
If $f(x) \approx \sum_k c_k p_k(x)$ on a domain enclosing the eigenvalues of $A$, then $f(A) \approx \sum_k c_k p_k(A)$. 
Each term in the Fourier series and Chebyshev polynomial 
can be performed efficiently using the Hamiltonian simulation algorithm \cite{BCK2015} and a quantum walk \cite{BC2012,CKS2017}, respectively. 
When we obtain the expansion, we can perform the matrix functions using a framework that implements the linear combination of unitaries (LCU) \cite{CW2012,BCCKS2015}. 
When the number of terms is proportional to $\poly\log(1/\epsilon)$, 
then the matrix function is performed efficiently.  
Unlike the method based on eigenvalue estimation, the LCU-based approach is performed efficiently, that is, the overall cost is $\poly\log(1/\epsilon)$.

As seen from the above, 
the QPE and LCU techniques are based on eigenvalue decomposition. 
An approach using another matrix decomposition, i.e., the singular value decomposition, was also considered \cite{GSLW2018,GSLW2019}. 
The method in \cite{GSLW2018,GSLW2019} is called as ``quantum singular value transformation'' (QSVT). 
As its name suggests, 
QSVT implements a (generalized) matrix function $f^\diamond (A) = Uf(\Sigma) V^\dag$ on quantum computers, 
where $A = U\Sigma V^\dag$ is the singular value decomposition.   
Roughly speaking, 
QSVT treats the matrix function corresponding to the top-left element of a $2 \times 2$ matrix 
$\e^{\im \phi_0 \sigma_z} \e^{\im \theta \sigma_x}\e^{\im \phi_1 \sigma_z} \e^{\im \theta \sigma_x}\e^{\im \phi_2 \sigma_z} \dots$, 
where $\theta = \arccos(x)$, and the scalar sequence $\Phi = (\phi_0, \phi_1,\dots,\phi_{d-1}) \in \R^d$ 
corresponds to the function to be implemented. 
When the function satisfies certain assumptions, 
the scalar sequence $\Phi = (\phi_0, \phi_1,\dots,\phi_{d-1}) \in \R^d$ exists. 
Additionally, classical numerical methods for computing the sequence efficiently have been proposed \cite{Haa2019,DMWL2020,CDGHS2020}.
If we compute  the sequence by preprocessing, 
then the circuit that consists of a block-encoding of $A$ and the rotation gates represents a block-encoding of the matrix functions. 
We do not use QSVT directly, but utilize it to perform the inverse. 
QSVT is based on qubitization \cite{LC2019} and quantum signal processing \cite{LYC2016}. 
For the details of QSVT and related techniques, see \cite{GSLW2018,GSLW2019} and \cite{MRTC2021}.

The quantum algorithms as mentioned above use the decompositions of matrix functions and the associated approximations. 
As we stated earlier, the matrix functions can be represented as contour integrals and can be approximated by the linear combination of inverses. 
A quantum algorithm that uses this integral representation was proposed in \cite{TOSU2020} for the case in which the contour is a circle centered at the origin. 
In \cite{TOSU2020}, the computing the state $\ket{f}$ is considered under the assumption that we have oracles to return elements and nonzero positions of $A$. 
To compute the approximation (of the matrix function) as the linear combination of inverses, 
a block-diagonal matrix and a circuit to multiply the coefficients are mainly considered. 
Here, diagonal blocks of this block-diagonal matrix are matrices we want to invert. 
In \cite{TOSU2020}, LCU-based quantum linear systems solver \cite{CKS2017} is used to invert the block-diagonal matrix. 
To apply this solver, we need oracles for the matrix elements and nonzero positions. 
Thus, \cite{TOSU2020} constructs an arithmetic circuit to access the block-diagonal matrix and so the procedure is intricate. 
The problem setup in this paper has slightly different from the problem in the previous one. 
This paper considers the construction of the block-encoding of the matrix function under the assumption that we have a block-encoding of $A$ 
(Note that the block-encoding of $A$ can be constructed by the oracles for $A$ \cite{GSLW2018}). 
But the strategy of this paper is similar to the previous one, 
that is, we consider the block-diagonal matrix and inverse of it. 
We use QSVT to obtain the block-encoding of the inverse. 
From the simplicity of the block-encoding framework, the method is also simple construction.

In the recent study \cite[Section 5.1]{TAWL2020}, a method for matrix functions based on integrals has been proposed. 
They consider a block-encoding of an inverse of a block-diagonal matrix and construct a block-encoding 
of a linear combination of the form $\sum_k w_k(z_kI - A)^{-1}$ by the similar circuit of LCU. 
We describe the difference in methods below. 
First, we consider general matrix functions, while they mainly consider the matrix exponential function. 
Second, their method uses Gauss–Legendre quadrature.  
On the other hand, we  consider a method based on the trapezoidal rule and an adaptable method for the matrix function approximations by various quadratures. 
Specifically, approximation by trapezoidal rule is discussed in \Cref{subsec:matfunc-circle}, 
and an approximation of the form $\sum_k w_k(z_kI + y_kA)^{-1}$ is discussed in \Cref{subsec:notcircle}. 
The form $\sum_k w_k(z_kI + y_kA)^{-1}$ arises in various approximations of the matrix function, 
for example, Gauss–Legendre quadrature, the double exponential formula \cite{TM74}. 
Third, our algorithm describes the construction of it more explicitly. 
Further, the cost and parameter (in the block-encoding) of the method in \Cref{subsec:notcircle} are described in detail. 

In \cite{TAWL2020}, another interesting formulation of the matrix function has been discussed. Specifically, it is to use a relation $f(y^{-1}) = g(y)$. 
They considered an approximation of $g$ by the Chebyshev polynomials and proposed a quantum algorithm using QSVT. 

\subsection{Overview and notations}
The paper is structured as follows. 
In \Cref{sec:tools}, we review the definitions of the block-encoding and the state-preparation pair, 
which is a pair of unitary operations for representing the coefficients of the linear combination of a matrix. 
We also describe some basic matrix operations. 
These discussions are used to construct the block-encodings of the block-diagonal matrix and its inverse. 
In \Cref{sec:MainAlgorithm}, we construct a framework for implementing a linear combination of the inverse of each diagonal block in the block-diagonal matrix 
under a given block-encoding of it and state-preparation pair with respect to the coefficients. 
\Cref{sec:MatrixFunction} describes the use of this framework and shows quantum algorithms for the matrix functions in terms of the block-encoding. 
Finally, we conclude this paper in \Cref{sec:conclusion}. 

We introduce several notations.  
We say that $A$ is an $n$-qubit matrix if the size of $A$ is $2^n \times 2^n$. 
$I_n$ denotes the $n$-qubit identity matrix. Throughout this paper, $M=2^m$, where $m$ is a positive integer. 
For $d_0,d_1,\dots,d_{M-1} \in \C$ and matrices $B_0,B_1,\dots,B_{M-1} \in \C^{N \times N}$, 
$\diag(d_0,d_1,\dots,d_{M-1})$ and $\diag(B_0,B_1,\dots,B_{M-1})$ denote a diagonal matrix and a block diagonal matrix, 
respectively. $\vec{0}$ denotes the zero vector. 
$\sigma_x, \sigma_y$ and $\sigma_z$ are the Pauli matrices. 
The swap operation that swaps $i$-th and $j$-th qubit is denoted as $\SWAP_i^j$. 
Let $\SWAP_{a,b} = \prod_{i=1}^a \SWAP_{b+i}^i$. 
For simplicity, in the analysis of the complexity, we count the uses of $U^\dag$ operation, the controlled $U$ operation, and the controlled $U^\dag$ operation as the unitary operation $U$. 
$\norm{ \cdot }$ denotes the spectral norm and $\norm{\cdot}_1$ denotes $1$--norm.  
We denote $\ket{0^a}$ by $\ket{0^a} = \ket{0}^{\otimes a} = \underbrace{\ket{0} \otimes \dots \otimes \ket{0}}_{a \textrm{ times} }$.

\section{Preliminaries}\label{sec:tools}
In this section, we review the block-encoding described in \cite{SGJ19,GSLW2018}. 
Block-encoding is a good framework for implementing matrix arithmetic on quantum computers. 
Using the block-encoding framework, 
we can construct a quantum algorithm to perform matrix arithmetic by combining basic operations and given block-encodings.
Furthermore, we can derive the cost, error, and number of ancilla qubits from the parameters in the given block-encodings. 

In the later section, to construct the block-encoding of (the approximation of) the matrix functions, 
we consider a block-encoding matrix $\sfA$. 
The content in this section is used to obtain the block-encodings of this block-diagonal matrix $\sfA$ and its inverse $\sfA^{-1}$.

\subsection{Block-encoding}\label{subsec:defofblockencoding}

In quantum algorithms for linear algebra, it is often necessary to consider a unitary operation $U$ that implements a non-unitary matrix $A$
such that $U\ket{0^a}\ket{\psi} \approx \frac{1}{\alpha} \ket{0^a}A\ket{\psi} + \ket{\Psi_{\perp}}$ on quantum computers, 
where $\ket{\psi}$ is an $n$-qubit state and $\ket{\Psi_{\perp}}$ is the state such that $(\gaiseki{0^a}{0^a} \otimes I_n)\ket{\Psi_{\perp}} = \vec{0}$. 
If we have the unitary operation $U$, then we can obtain a state close to $A\ket{\psi}/\norm{A\ket{\psi}}$ with high probability 
by using amplitude amplification \cite{BHMT2002}. 
Focusing on the structure of the unitary operation $U$, one can see that the top-left block closes to $A/\alpha$, that is, 
\begin{align}
  U \approx \begin{bmatrix}
  A/\alpha & \cdot \\
  \cdot     & \cdot
\end{bmatrix},
\end{align}
where $\cdot$ represents other matrices. 
Block-encoding enacts such a unitary operation $U$ on the matrix $A$ using a positive scalar $\alpha$ as a scaling factor, 
the number of ancilla qubits $a$, and the error factor $\epsilon$. 
The definition of block-encoding is as follows. 

\begin{defi}[Block-encoding {\cite[Definition 43]{GSLW2018}}]
  Suppose that $A$ is an $n$-qubit matrix, $\alpha, \epsilon \in \R_+$, and $a \in \N$. 
  Then, we say that the $n+a$-qubit unitary operation $U$ is the $(\alpha,a,\epsilon)$-block-encoding of $A$ if
  \begin{align}
    \bigBnorm{A-\alpha(\bra{0^a} \otimes I_n)U(\ket{0^a} \otimes I_n)} \le \epsilon. 
    \label{eqdefi:be}
  \end{align}  
\end{defi}

For simplicity, we assume that $\norm{A} \le \alpha$ for the block-encoded matrix in a given block-encoding.
As implementation of the block-encoding of $A$ on quantum computers, some ways have been considered. 
When the oracle for the elements and the oracle for the nonzero positions of $A$ are given, 
the quantum circuit of the block-encoding of a sparse matrix $A$ can be performed efficiently\cite[Lemma 48]{GSLW2018}.  
In addition, if $A$ is described as a linear combination of unitaries (such as Pauli-matrices), 
then the LCU circuit on it becomes the block-encoding of $A$ (see also, \cite[Section II.A]{DMWL2020}).
From this background, in this study, we assume that the block-encoding of matrix $A$ is given. 
Therefore, in the analysis of the complexity, we mainly focus on the number of uses of the block-encoding of matrix $A$ and other one- and two-qubit gates. 

We mention to same of clear properties of block-encodings in the below. 
If $U$ is an $(\alpha, a, \epsilon)$-block-encoding of $A$, then $U$ is also the $(c\alpha, a, c\epsilon)$-block-encoding of $cA$, where $c$ is a positive constant. 
If a block-encoded matrix is a unitary, then the block-encoding is called trivial block-encoding. 
When $U$ is $(\alpha, a, \epsilon)$-block-encoding of $A$, $I_b \otimes U \ (b \ge 0)$ is $(\alpha, a+b, \epsilon)$-block-encoding of $A$.

\subsection{State preparation pair}

We show a pair of unitary operations, which is state preparation method, 
for representing the coefficients of the given linear combination of matrices. 
This ``state-preparation pair'' is used in block-encoding for the linear combination of matrices. 
The state-preparation pair is defined as follows.

\begin{defi}[State-preparation pair {\cite[Definition 51]{GSLW2018}}]
  Let $m$ be a positive integer and let $t$ be an integer such that $t \le  2^m-1$. 
  Suppose that $\vec{v} = (v_0, v_1, \dots, v_{t-1}) \in \C^t$ and $\norm{\vec{v}}_1 \le \mu$. 
  We say that the pair of unitaries $(V_L, V_R)$ is the $(\mu, m, \delta)$-state-preparation pair for a nonzero vector $\vec{v}$ if 
  $V_L\ket{0^m} = \sum_{j=0}^{2^m-1}c_j\ket{j}$ and $V_R\ket{0^m} =\sum_{j=0}^{2^m-1}d_j\ket{j}$ 
  such that $\sum_{j=0}^{t-1}\abs{\mu(c_j^\ast d_j) - v_j} \le \delta$ and $c_j^\ast d_j = 0$ holds for all $j=t,t+1,\dots,2^m-1$. 
\end{defi}


For example, a pair of unitaries $(V_L, V_R)$ such that $V_L\ket{0^m} = \sqrt{\norm{\vec{v}}_1^{-1}}\sum_j \sqrt{v_j^\ast}\ket{j}$ and $V_R\ket{0^m} = \sqrt{\norm{\vec{v}}_1^{-1}}\sum_j \sqrt{v_j}\ket{j} $ 
for an nonzero vectors $\vec{v} \in \C^{2^m}$ is a $(\norm{\vec{v}}_1, m, 0)$-state-preparation-pair for the vector $\vec{v}$. 
We can implement such unitary operations with $O(2^m)$ gates at least using the method in \cite{SBM2006}. 
In this paper, for simplicity, we assume that the state-preparation-pair for a given nonzero vector can be implemented without error by using the method in \cite{SBM2006}.

\subsection{Basic operation}\label{subsec:lincomb-prod-tensor}
In this subsection,
we explain the block-encodings of the product, tensor product, and linear combination of matrices.
In the last of this subsection (\Cref{subsubsec:BELCTensor}),
we construct the block-encoding on the linear combination of the tensor products by combining the block-encodings on the tensor product and the linear combination.
In particular, this block-encoding is used to construct the block-encoding of the block-diagonal matrix that is constructed from the approximation of matrix functions.

\subsubsection{Product}\label{subsubsec:product}
When the block-encodings of $A$ and $B$ are given,
the block-encoding of the product $AB$ of these matrices can be constructed as the product of the given block-encodings \cite{GSLW2018}.

\begin{prop}[Product of two block-encoded matrices {\cite[Lemma 53]{GSLW2018}}] \label{prop:productbe}
  If $U_A$ is an $(\alpha, a, \epsilon_A)$-block-encoding of an $n$-qubit matrix $A$ and
  $U_B$ is a $(\beta,  b, \epsilon_B)$-block-encoding of an $n$-qubit matrix $B$, then
  $(I_b\otimes U_A)(I_a \otimes U_B)$ is an $(\alpha\beta, a+b, \alpha\epsilon_B + \beta\epsilon_A)$-block-encoding of $AB$.
\end{prop}

\begin{proof}
  See \cite[Lemma 53]{GSLW2018}.
\end{proof}

Note that the meaning of notation of the tensor (Kronecker) product in the above proposition is different from
the normal notation.
Namely, $I_a$ and $I_b$ act on the ancilla qubits of $U_A$ and $U_B$, respectively.
When we write the equation using the normal notation of tensor products,
we have to insert the swap gates among the operation of the block-encodings.
We describe this in the following. This is based on the description in \cite{CB2020}.
Let $\text{SWAP}_i^j$ be the swap operation for $i$-th and $j$-th qubit and let
$\text{SWAP}_{a, b} := \prod_{i=1}^a \text{SWAP}_{b+i}^i$.
This swap gates satisfies
\begin{align}
\text{SWAP}_{a, b}
\Bigl( \ket{0^a} \otimes I_b \Bigr)
&= \prod_{i=1}^a \text{SWAP}_{b+i}^i \Bigl( \ket{0^a} \otimes I_b \Bigr) \notag \\
&= \prod_{i=1}^{a-1} \text{SWAP}_{b+i}^i \Bigl(\ket{0^{a-1}} \otimes I_b \otimes \ket{0} \Bigr) \notag \\
&= I_b \otimes \ket{0^a}.
\end{align}
Thus,
\begin{align}
  &
  \Bigl( \bra{0^{a+b}} \otimes I_n \Bigr)
  ( I_b \otimes U_A )
  (\text{SWAP}_{a, b} \otimes I_n)
  ( I_a \otimes U_B )
  \Bigl( \ket{0^{b+a}} \otimes I_n \Bigr) \notag \\
  &=
  \Bigl\{  \bra{0^b} \otimes \bigl( (\bra{0^a} \otimes I_n)U_A \bigr) \Bigr\}
  (\text{SWAP}_{a, b} \otimes I_n)
  \Bigl\{ \ket{0^a} \otimes \bigl( U_B(\ket{0^b} \otimes I_n)  \bigr) \Bigr\} \notag \\
  &=
  \Bigl( (\bra{0^a} \otimes I_n)U_A \Bigr)
  \Bigl(\bra{0^b} \otimes I_{a + n} \Bigr)
  (\text{SWAP}_{a, b} \otimes I_n)
  \Bigl(\ket{0^a} \otimes I_{b + n} \Bigl)
  \Bigl( U_B(\ket{0^b} \otimes I_n) \Bigr) \notag \\
  &=
  \Bigl( (\bra{0^a} \otimes I_n)U_A(\ket{0^a} \otimes I_n) \Bigr)
  \Bigl( (\bra{0^b} \otimes I_n)U_B(\ket{0^b} \otimes I_n) \Bigr),
\end{align}
where we used
$(\bra{0^b} \otimes I_{a + n})(\text{SWAP}_{a, b} \otimes I_n)(\ket{0^a} \otimes I_{b+n})
= (\bra{0^b} \otimes I_{a+n})(I_b \otimes \ket{0^a} \otimes I_n)
= (\ket{0^a} \otimes I_n)(\bra{0^b} \otimes I_n) $.
Thus, we can see that $( I_b \otimes U_A )(\text{SWAP}_{a, b} \otimes I_n)( I_a \otimes U_B )$ is the block-encoding of $AB$.
The circuit of the block-encoding for the product is shown in \Cref{qc:BEproduct}.

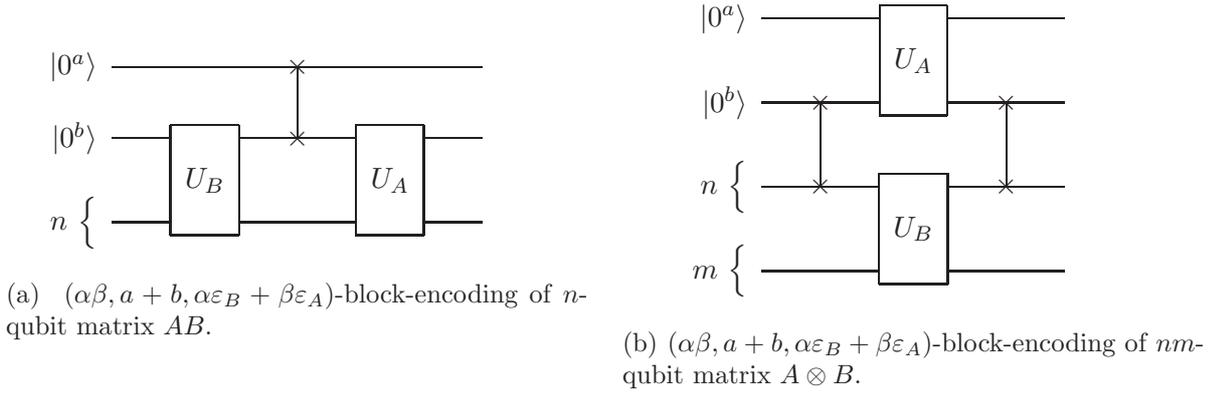
\begin{figure}[htbp]
  \begin{tabular}{cc}
    \begin{minipage}{0.45\hsize}
      \centering
      \[
        \Qcircuit @C=2em @R=2em {
        \lstick{\ket{0^a}}     & \qw                & \qswap     & \qw                & \qw \\
        \lstick{\ket{0^b}}     & \multigate{1}{U_B} & \qswap\qwx & \multigate{1}{U_A} & \qw \\
        \lstick{n \  \Bigl\{ } &  \ghost{U_B}        & \qw        & \ghost{U_A}        & \qw 
        }
      \]
      \subcaption{
        \ $(\alpha\beta, a+b, \alpha\epsilon_B + \beta\epsilon_A)$-block-encoding of $n$-qubit matrix $AB$.
      }
      \label{qc:BEproduct}
    \end{minipage}
    \hspace{2mm}
    \begin{minipage}{0.45\hsize}
      \centering
        \[
        \Qcircuit @C=2em @R=2em {
        \lstick{\ket{0^a}} & \qw        & \multigate{1}{U_A} & \qw    & \qw \\
        \lstick{\ket{0^b}} & \qswap     & \ghost{U_A}        & \qswap & \qw \\
        \lstick{n \ \Bigl\{ }         & \qswap\qwx & \multigate{1}{U_B} & \qswap\qwx & \qw \\
        \lstick{m \ \Bigl\{ }         & \qw        & \ghost{U_B}        & \qw    & \qw
        }
        \]
        \subcaption{
          $(\alpha\beta, a+b, \alpha\epsilon_B + \beta\epsilon_A)$-block-encoding of $nm$-qubit matrix $A \otimes B$.
        }
        \label{qc:BEtensorproduct}
    \end{minipage}
  \end{tabular}
  \caption{
    Block encodings for \textit{product} and \textit{tensor product} of $A$ and $B$,
    where $U_A$ is an $(\alpha, a, \epsilon_A)$-block-encoding of an $n$-qubit matrix $A$ and
    $U_B$ is a $(\beta, b, \epsilon_B)$-block-encoding of an $m$-qubit matrix $B$ ($n = m$ for (a)).
    The circuits (a) and (b) is related to \Cref{prop:productbe,prop:tensorproductbe}, respectively.
  }
\end{figure}

\subsubsection{Tensor product}
Similar to the case of the product,
when the block-encodings of $A$ and $B$ are given,
the block-encoding of the tensor (Kronecker) product $A \otimes B$ of these matrices can be constructed as the
tensor product of the given block-encodings.

\begin{prop}[\textit{Tensor} product of block-encoded matrices]\label{prop:tensorproductbe}
  If $U_A$ is an $(\alpha, a, \epsilon_A)$-block-encoding of an $n$-qubit matrix $A$ and
  $U_B$ is a $(\beta,  b, \epsilon_B)$-block-encoding of an $m$-qubit matrix $B$, then
  $(I_b \otimes (U_A \otimes I_m))(I_a \otimes (I_n \otimes U_B))$ is
  an $(\alpha\beta, a+b, \alpha\epsilon_B + \beta\epsilon_A)$-block-encoding of $A\otimes B$.
\end{prop}

\begin{proof}
  The following shows that $U_A \otimes I_m$ is the $(\alpha, a, \epsilon_A)$-block-encoding of $A \otimes I_m$:
  \begin{align}
    &\Biggl\|
      A \otimes I_m - \alpha (\bra{0^a} \otimes I_n \otimes I_m)(U_A \otimes I_m)(\ket{0^a} \otimes I_n \otimes I_m)
    \Biggr\| \\
    &= \Biggl\|
      A \otimes I_m - \alpha (\bra{0^a} \otimes I_n)U_A(\ket{0^a}\otimes I_n)\otimes I_m
    \Biggr\| \le \epsilon_A.
  \end{align}
  Similarly, $I_n \otimes U_B$ is the $(\beta, b, \epsilon_B)$-block-encoding of $I_n \otimes B$.
  Thus, from $A \otimes B = (A \otimes I_m)(I_n \otimes B)$ and \Cref{prop:productbe},
  $(I_b \otimes (U_A \otimes I_m) )(I_a \otimes (I_n \otimes U_B) )$ is
  the $(\alpha\beta, a+b, \alpha\epsilon_B+\beta\epsilon_A)$-block-encoding of $A\otimes B$.
\end{proof}
The same discussion when there is no assumption on the norm of the block-encoded matrices ($\norm{A} \le \alpha$) was described in  \cite{CB2020}.
Similarly to the discussion in \Cref{subsubsec:product},
the meaning of the tensor products in the above proposition is different from the normal notation.
When we use the normal notation, we have to insert the swap gates before and after $U_A \otimes U_B$.
Let $\text{SWAP}_{b, n} = \prod_{i=1}^b \text{SWAP}_{i+n}^i$.
Then, $\text{SWAP}_{b, n} \ket{0^b} \otimes I_n = I_n \otimes \ket{0^b}$ holds. Thus,
\begin{align}
  &\left( \bra{0^{a+b}} \otimes I_{n+m} \right)
    \left( I_a \otimes \text{SWAP}_{b, n} \otimes I_m \right)\left(U_A \otimes U_B \right)\left( I_a \otimes \text{SWAP}_{b, n} \otimes I_m \right)
    \left( \ket{0^{a+b}} \otimes I_{n+m} \right) \notag \\
  &=
  \left( \bra{0^a} \otimes I_n \otimes \bra{0^b} \otimes I_m \right)
  \left(U_A \otimes U_B \right)
  \left( \ket{0^a} \otimes I_n \otimes \ket{0^b} \otimes I_m \right)  \notag \\
  &=
  ( \bra{0^a} \otimes I_n )U_A(\ket{0^a} \otimes I_n )
  \otimes
  ( \bra{0^b} \otimes I_m) U_B(\ket{0^b} \otimes I_m).
\end{align}
Therefore, we can see that $\left( I_a \otimes \text{SWAP}_{b, n} \otimes I_m \right)\left(U_A \otimes U_B \right)\left( I_a \otimes \text{SWAP}_{b, n} \otimes I_m \right)$
is the block-encoding of $A \otimes B$.
The circuit of the block-encoding for the tensor product is shown in \Cref{qc:BEtensorproduct}.

\subsubsection{Linear combination}
Suppose that the block-encodings of the matrices $A_0,A_1,\dots,A_{M-1}$ are given. Then, a block-encoding
of the linear combination of these matrices can be constructed \cite[Lemma 52]{GSLW2018} by a circuit similar to a well-known LCU framework \cite{CW2012,BCCKS2015}.
However, the method in \cite[Lemma 52]{GSLW2018} assumes
that the scaling factors (denoted by $\alpha_0, \alpha_1,\dots,\alpha_{M-1}$ in the following) of the block-encodings are equal.
For our purpose, in the following proposition, we describe a method in which such a restriction is removed by adjusting the state-preparation pair.
The circuit of the block-encoding for the linear combination is shown in \Cref{qc:BELC}.

\begin{prop}[Linear combination of block-encoded matrices]\label{prop:linearcombination}
  Let $A_j$ be an $n$-qubit matrix and let $A = \sum_{j=0}^{t-1} y_j A_j$, where $y_0, y_1, \dots, y_{t-1}\in \C$.
  If $U_j$ is an $(\alpha_j, a_j, \epsilon_j)$-block-encoding of  $A_j$ and
  $(P_L, P_R)$ is a $(\beta, b, 0)$-state-preparation pair of $(y_0\alpha_0, y_1\alpha_1, \dots, y_{y-1}\alpha_{t-1})$,
  then
  \begin{align}
    \widetilde{W} =
    \Biggl(P_L^\dag \otimes I_{a+n}\Biggr)
    \left( \sum_{j=0}^{t-1}\gaiseki{j}{j} \otimes (I_{a - a_j} \otimes U_j) + \sum_{j=t}^{2^b - 1} \gaiseki{j}{j}\otimes I_{a+n} \right)
    \Biggl(P_R \otimes I_{a+n}\Biggr)
  \end{align}
  is
  a $(\beta, a+b, \sum_{j=0}^{t-1}\abs{y_j}\epsilon_j )$-block-encoding of $A = \sum_{j=0}^{t-1} y_j A_j$,
  where $a = \max_j a_j$.
\end{prop}

\begin{proof}
  Let $V_j = I_{a - a_j} \otimes U_j$.
  Then, $V_j$ is the $(\alpha_j, a, \epsilon_j)$-block-encoding of $A_j$.
  According to the definition of block-encoding, we have
  \begin{align}
    \bigCnorm{
     A - \beta (\bra{0^b} \otimes \bra{0^a} \otimes I_n ) \widetilde{W} (\ket{0^b} \otimes \ket{0^a} \otimes I_n )
    }
    &=
    \bigCnorm{
      A - \sum_{j=0}^{t-1} y_j \alpha_j (\bra{0^a} \otimes I_n)V_j(\ket{0^a} \otimes I_n)
     } \notag \\
     &=
     \bigCnorm{
       \sum_{j=0}^{t-1} y_j \bigl( A_j - \alpha_j (\bra{0^a} \otimes I_n)V_j(\ket{0^a} \otimes I_n) \bigr)
      } \notag \\
    &\le
    \sum_{j=0}^{t-1} \abs{y_j}  \bigCnorm{ A_j - \alpha_j (\bra{0^a} \otimes I_n)V_j(\ket{0^a} \otimes I_n) }
      \notag \\
    &\le \sum_{j=0}^{t-1}\abs{y_j}\epsilon_j.
  \end{align}
\end{proof}

\begin{figure}[htbp]
  \begin{center}
  \[
    \Qcircuit @C=1em @R=1em {
    \lstick{\ket{0}}   & \multigate{2}{P_R} & \ctrl{1} & \ctrl{1}  & \ctrl{1}  &\qw  & \cdots & &\qw & \multigate{2}{P_L^\dag} & \qw \\
    \lstick{\ket{0}}   & \ghost{P_R}        & \ctrl{1} & \ctrl{1}  & \ctrlo{1}&\qw & \cdots & &\qw & \ghost{P_L^\dag}         & \qw \\
    \lstick{\ket{0}}   & \ghost{P_R}        & \ctrl{1} & \ctrlo{1} & \ctrl{1}&\qw  & \cdots & &\qw & \ghost{P_L^\dag}         & \qw \\
    \lstick{\ket{0^a}} & \qw                & \multigate{1}{U_0} & \multigate{1}{U_1} & \multigate{1}{U_2}&\qw & \cdots & & \qw & \qw & \qw  \\
    \lstick{n \ \Bigl\{ }         & \qw                & \ghost{U_0}        & \ghost{U_1}        & \ghost{U_2}& \qw       & \cdots & & \qw & \qw & \qw
    }
  \]
  \end{center}
\caption{
  $(\beta, a+b, \sum_{j=0}^{t-1}\abs{y_j}\epsilon_j)$-block-encoding of $n$-qubit matrix $A = \sum_{j=0}^{t-1} y_j A_j$,
  where
  $(P_L,P_R)$ is a $(\beta, b, 0)$-state-preparation pair for $(y_0\alpha_0, \dots, y_{t-1}\alpha_{t-1})$,
  where $U_j$ is the $(\alpha_j, a, \epsilon_j)$-block-encoding of $A_j$.
  This circuit is related to \Cref{prop:linearcombination}.
}
\label{qc:BELC}
\end{figure}
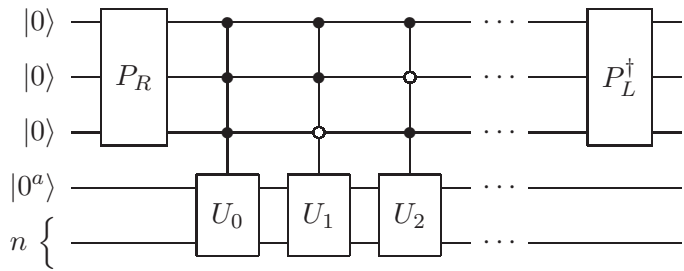

\subsubsection{Linear combination of tensor products}\label{subsubsec:BELCTensor}

We describe the block-encoding of the linear combination of tensor products of matrices in the following.
As mentioned in the beginning of this subsection, we use the following to construct the block-encoding of the block-diagonal matrix in the later section.
Specifically, we consider the block-diagonal matrix from the approximation of matrix functions
and then we show that it is decomposed to the linear combination of the tensor products.
After that, using the following proposition, we obtain the block-encoding of the block-diagonal matrix.
The circuit of the block-encoding for the linear combination is shown in \Cref{qc:BELCTensor}.

\begin{prop}[Linear combination of tensor products of block-encoded matrices]\label{prop:linearcobination_tensorproduct}
  Let $A_j$ be an $n$-qubit matrix and let $B_j$ be an $m$-qubit matrix.
  Let $C = \sum_{j=0}^{t-1} y_j (A_j \otimes B_j)$, where $y_0, y_1, \dots, y_{t-1}\in \C$.
  If
  $U_j$ is an $(\alpha_j, a_j, \epsilon_j)$-block-encoding of  $A_j$,
  $V_j$ is a $(\beta_j, b_j, \delta_j)$-block-encoding of  $B_j$,
  and
  $(P_L, P_R)$ is an $(\gamma, c, 0)$-state-preparation pair of $(y_0\alpha_0\beta_0, y_1\alpha_1\beta_1, \dots, y_{y-1}\alpha_{t-1}\beta_{t-1})$,
  then
  \begin{align}
    \widetilde{W}
    &= \Biggl(P_L^\dag \otimes I_a \otimes \SWAP_{b, n} \otimes I_m \Biggr) \notag \\
    &\qquad \left( \sum_{j=0}^{t-1}\gaiseki{j}{j} \otimes (I_{a - a_j} \otimes U_j) \otimes (I_{b - b_j} \otimes V_j) + \sum_{j=t}^{2^b - 1} \gaiseki{j}{j} \otimes I_{a+n+m} \right) \notag \\
    &\qquad\quad \Biggl(P_R \otimes I_a \otimes \SWAP_{b, n} \otimes I_m \Biggr)
  \end{align}
  is
  a $(\gamma, a+b+c, \sum_{j=0}^{t-1}\abs{y_j}(\alpha_j\delta_j + \beta_j\epsilon_j) )$-block-encoding of $C = \sum_{j=0}^{t-1} y_j (A_j \otimes B_j)$,
  where $a = \max_j a_j, b = \max_j b_j$, and $\SWAP_{b, n} = \prod_i^b \SWAP_{i+n}^i$.
\end{prop}

\begin{proof}
  We write
  \begin{align}
    W &=
    \sum_{j=0}^{t-1}\gaiseki{j}{j} \otimes
    \underbrace{
    \Bigl(I_a \otimes \SWAP_{b, n} \otimes I_m \Bigr)\Bigl(I_{a - a_j} \otimes U_j \otimes I_{b - b_j} \otimes V_j \Bigr) \Bigl(I_a \otimes \SWAP_{b, n} \otimes I_m \Bigr)
     }_{=:U_{A_j \otimes B_j}}
      \notag \\
    &\qquad + \sum_{j=t}^{2^b - 1} \gaiseki{j}{j} \otimes I_{a+n+m}.
  \end{align}
  Then, we have $\widetilde{W} = ( P_L^\dag \otimes I_{a+b+n+m} ) W  (P_R \otimes I_{a+b+n+m} )$.
  From \Cref{prop:tensorproductbe},
  $U_{A_j \otimes B_j} = \Bigl(I_a \otimes \SWAP_{b, n} \otimes I_m \Bigr)\Bigl(I_{a - a_j} \otimes U_j \otimes I_{b - b_j} \otimes V_j \Bigr) \Bigl(I_a \otimes \SWAP_{b, n} \otimes I_m \Bigr)$
  is a $(\alpha_j\beta_j, a+b, \alpha_j\delta_j + \beta_j\epsilon_j)$-block-encoding of $A_j \otimes B_j$.
  Thus, the proposition follows from \Cref{prop:linearcombination}.
\end{proof}

\begin{figure}[htbp]
    \[
      \Qcircuit @C=0.7em @R=1em {
      \lstick{\ket{0}}   & \multigate{2}{P_R} & \ctrl{1}           & \ctrl{1}           & \ctrl{1}           & \ctrl{1}           & \ctrl{1}           & \ctrl{1}           & \qw & \cdots & & \qw & \multigate{2}{P_L^\dag} & \qw \\
      \lstick{\ket{0}}   & \ghost{P_R}        & \ctrl{1}           & \ctrl{1}           & \ctrl{1}           & \ctrl{1}           & \ctrlo{1}          & \ctrlo{1}          & \qw & \cdots & & \qw & \ghost{P_L^\dag}        & \qw \\
      \lstick{\ket{0}}   & \ghost{P_R}        & \ctrl{1}           & \ctrl{3}           & \ctrlo{1}          & \ctrlo{3}          & \ctrl{1}           & \ctrl{3}           & \qw & \cdots & & \qw & \ghost{P_L^\dag}        & \qw \\
      \lstick{\ket{0^a}} & \qw                & \multigate{1}{U_0} & \qw                & \multigate{1}{U_1} & \qw                & \multigate{1}{U_2} & \qw                & \qw & \cdots & & \qw & \qw        & \qw  \\
      \lstick{\ket{0^b}} & \qswap             & \ghost{U_0}        & \qw                & \ghost{U_1}        & \qw                & \ghost{U_2}        & \qw                & \qw & \cdots & & \qw & \qswap     & \qw  \\
      \lstick{n \ \bigl\{} & \qswap\qwx         & \qw                & \multigate{1}{V_0} & \qw                & \multigate{1}{V_1} & \qw                & \multigate{1}{V_2} & \qw & \cdots & & \qw & \qswap\qwx & \qw  \\
      \lstick{m \ \bigl\{} & \qw                & \qw                & \ghost{V_0}        & \qw                & \ghost{V_1}        & \qw                & \ghost{V_2}        & \qw & \cdots & & \qw & \qw        & \qw
      }
    \]
    \caption{
      $(\gamma, a+b+c, \sum_{j=0}^{t-1}\abs{y_j}(\alpha_j\delta_j + \beta_j\epsilon_j) )$-block-encoding of $C = \sum_{j=0}^{t-1} y_j (A_j \otimes B_j)$,
      where
      $U_j$ is an $(\alpha_j, a, \epsilon_j)$-block-encoding of an $n$-qubit matrix $A_j$,
      $V_j$ is an $(\beta_j, b, \delta_j)$-block-encoding of an $m$-qubit matrix $B_j$,
      and
      $(P_L, P_R)$ is an $(\gamma, c, 0)$-state-preparation pair of $(y_0\alpha_0\beta_0, y_1\alpha_1\beta_1, \dots, y_{y-1}\alpha_{t-1}\beta_{t-1})$.
      This circuit is related to \Cref{prop:linearcobination_tensorproduct}.
    }
    \label{qc:BELCTensor}
\end{figure}
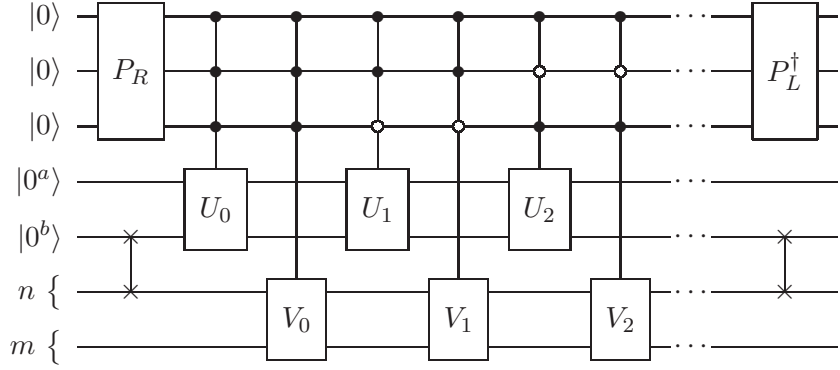

\subsection{Inverse}\label{subsec:inverse}
If we construct the block-encoding of a certain matrix, then we can construct a block-encoding of the inverse of the matrix using QSVT \cite{GSLW2018}. 
To invert the block-diagonal matrix, we use QSVT. 
In this subsection, we briefly describe the result of \cite{GSLW2018}. 
Here, we assume that $A$ is Hermitian, but such a construction can be removed by extending the matrix, as described in \cite{HHL2009}. 
In the next subsection, we describe this extension in detail. 
We first describe the proposition for the block-encoding of a given polynomial.

\begin{prop}[{\cite[Theorem 56]{GSLW2018}}]\label{prop:becircuit}
  Suppose that $U_A$ is an $(\alpha, a, \epsilon)$-block-encoding of an $n$-qubit Hermitian matrix $A$. 
  If $\delta > 0$ and $P \in \R[x]$ is a degree-$d$ polynomial satisfying that $\abs{P(x)} \le 1$ for all $x \in [-1, 1]$, 
  then, there is a quantum circuit $U'$, which is an $(1, a+2, 4d\sqrt{\frac{\epsilon}{\alpha} } + \delta)$ 
  -block-encoding of $\frac{1}{2}P(\frac{A}{\alpha})$, 
  and consists of $\order(d)$ uses of $U_A$ and $\order((a+1)d)$ other one- and two-qubit gates. 
  Moreover, we can compute a  description of such a circuit with a classical computer in $\poly(d, \log\frac{1}{\delta})$ time. 
\end{prop}

This proposition indicates that the complexity and the error increase as the degree increases. 
We can construct the block-encoding of the inverse using a polynomial that approximates $1/x$. 
In the following proposition, we show that there exists such a polynomial on a certain range. 

\begin{prop}[Polynomial approximation of $\frac{1}{x}$ {\cite[Corollary 69]{GSLW2018}}]\label{cor69inGSLW2018}
  Let $\delta, \sigma \in (0, \frac{1}{2}]$, 
  then there is an odd polynomial $P \in \R[x]$ of degree $O(\frac{1}{\sigma}\log\frac{1}{\delta})$ that is 
  $\delta$-approximating $f(x) = \frac{3}{4}\frac{\sigma}{x}$ on the domain $\calI = [-1, -\sigma] \cup [\sigma, 1]$, 
  moreover, it is bounded 1 in absolute value. 
\end{prop}

For the construction of the odd polynomial, see \cite{GSLW2018}. 
We say that $P$ is the $\delta$-approximation of $f(x)$ on the domain $\calI$ if 
$\abs{f(x) - P(x)} \le \delta$ holds for all $x \in \calI$. 
We summarize the above propositions in following to explicitly state on the block-encoding of the inverse of a Hermitian matrix. 

\begin{prop}[Block-encoding of the inverse]\label{prop:inversebe}
  Suppose that $U_A$ is an $(\alpha, a, \epsilon)$-block-encoding of an $n$-qubit Hermitian matrix $A$ 
  that its eigenvalues is on the range $\calI = [-\alpha, -1/\beta] \cup [1/\beta, \alpha]$. 
  Let $\delta \in (0, \frac{3}{4}]$. 
  Then, there is a quantum circuit $\tilde{U}$, which is $(\tilde{\beta}, a+2, \tilde{\epsilon})$-block-encoding of $A^{-1}$, 
  where 
  \begin{align}
    \tilde{\beta} = \frac{16}{3} \beta, \qquad 
    \tilde{\epsilon} = \left(4d\sqrt{\frac{\epsilon}{\alpha}} + \delta\right)\tilde{\beta}, 
  \end{align}
  and consists of $O(d)$ uses of $U_A$ and $O((a+1)d)$ other one- and two-qubit gates, 
  where 
  \begin{align}
    d = O\left(\alpha\beta \log\frac{1}{\delta} \right ).     
  \end{align}
  Moreover, we can compute a  description of such a circuit with a classical computer in $\poly(d, \log\frac{1}{\delta})$ time.   
\end{prop}

\begin{proof}
  We define $\bar{\kappa} = 2\alpha\beta$, such that $\bar{\kappa}$ satisfies $\frac{1}{\bar{\kappa}} \le \frac{1}{2}$. 
  Thus, from \cref{cor69inGSLW2018}, we can see that 
  there is an odd polynomial $P \in \R[x]$ of degree $d = O(\alpha\kappa \log\frac{1}{\delta_1})$ that is 
  a $\delta_1$-approximation of $f(x) = \frac{3}{4}\frac{1}{\bar{\kappa}x}$ on the domain $\calI = [-1,-1/\bar{\kappa}] \cup [1/\bar{\kappa}, 1]$. 
  Moreover, the absolute value of $P$ is bounded by 1. 
  Hence, it follows from \cref{prop:becircuit} that 
  we can construct a quantum circuit that is the 
  $(1,a+2, 4d\sqrt{\frac{\epsilon}{\alpha}}+\delta_2)$-block encoding of $\frac{1}{2}P(\frac{A}{\alpha})$ 
  and that consists of $O(d)$ uses of $U_A$ and $O((a+1)d)$ other one- and two-qubit gates. 
  Moreover, we can compute a  description of such a circuit with a classical computer in $\poly(d, \log(1/\delta_2))$ time.   
  By the definition of block-encoding, we have 
  \begin{align}
    \bigCnorm{ \frac{1}{2}P\l( \frac{A}{\alpha} \r) - (\bra{0^{a+2}} \otimes I_n )\tilde{U}(\ket{0^{a+2}} \otimes I_n) } \le 4d\sqrt{\frac{\epsilon}{\alpha}} + \delta_2. 
    \label{inq:inv2}
  \end{align}
  Note that the eigenvalues of the Hermitian matrix $\frac{1}{\alpha}A$ are in the domain $\calI$. 
  Therefore, 
  \begin{align}
    \bigCnorm{ P\left( \frac{A}{\alpha} \right) - f\left(\frac{A}{\alpha}\right)} = 
    \bigCnorm{ P\left( \frac{A}{\alpha} \right) - \frac{3}{4}\frac{1}{\bar{\kappa}} \left(\frac{A}{\alpha}\right)^{-1} } 
    = 
    \bigCnorm{ P\left( \frac{A}{\alpha} \right) - \frac{3}{8}\frac{1}{\beta} A^{-1} } 
    \le \delta_1. 
    \label{inq:inv1}W
  \end{align}
  Thus, $\tilde{U}$ is the $(1, a+2, 4d\sqrt{\epsilon/\alpha} + \delta_2 + \frac{1}{2}\delta_1)$-block-encoding of $\frac{3}{16}\beta^{-1} A^{-1} = \tilde{\beta}^{-1} A^{-1}$. 
  By setting $\delta = \delta_1 = \delta_2 < 1/2$ using $\delta \in (0, 1/2]$, we can see that 
  $\tilde{U}$ is the $(\tilde{\beta}, a+2, \tilde{\epsilon})$-block-encoding of $A^{-1}$. 
\end{proof}

Because we assumed that $\norm{A} \le \alpha$, the maximum absolute value of the eigenvalues is less than or equal to $\alpha$. 
If the value of $1/\beta$ equals to a minimum of eigenvalues in absolute value, then, $\alpha\beta \ge  \kappa$ and $\tilde{\beta} = (3/16)\kappa/\alpha$, 
where $\kappa$ is the condition number of $A$. 
This means that we need $\poly(\kappa)$ runtime to perform the inverse when we used QSVT.

\subsection{Extension}\label{subsec:be-extended}
When we construct the block-encoding of an inverse using QSVT, the matrix we want to invert must be Hermitian. 
This constraint can be removed by considering the extension of the matrix.  
Specifically, we consider an extended matrix defined as follows: 
\begin{align}
  \bar{A} := A \otimes \gaiseki{0}{1} + A^\dag \otimes \gaiseki{1}{0}. 
  \label{eq:extended_mat}
\end{align}
We can immediately see that this extended matrix is Hermitian. 
Additionally, the inverse of this extended matrix is $\bar{A}^{-1} = A^{-1} \otimes \gaiseki{1}{0} + (A^\dag)^{-1} \otimes \gaiseki{0}{1}$. 
The block-diagonal matrix that constructed from the approximation of the matrix function is not Hermitian. Therefore, to invert it using QSVT, we need to consider this extension. 
The eigenvalues of $\bar{A}$ are $\pm \sigma_j \ (j = 1,2,\dots,N)$, where $\sigma_j$ are the singular values of $A$. 
Hence, we have $\norm{\bar{A}} = \norm{A}$. 
In the following, we construct a block-encoding of $\bar{A}$ using the $\sigma_x$ operation and the block-encoding of $A$.  
The circuit is shown in \Cref{qc:extendedbe}.

\begin{prop}\label{prop:extendedbe}
  If $U$ is an $(\alpha,a,\epsilon)$-block-encoding of an $n$-qubit matrix $A$, 
  then $(U \otimes \gaiseki{0}{0} + U^\dag \otimes \gaiseki{1}{1})(I_{a+n} \otimes \sigma_x)$ is 
  an $(\alpha, a, 2\epsilon)$-block-encoding of an extended matrix $\bar{A} = A \otimes \gaiseki{0}{1} + A^\dag \otimes \gaiseki{1}{0}$. 
\end{prop}
\begin{proof}
  We have
  \begin{align}
    & \bigCnorm{
      \bar{A} - \alpha \Bigl(\bra{0^a} \otimes I_{n+1} \Bigr) \Bigl(U \otimes \gaiseki{0}{0} + U^\dag \otimes \gaiseki{1}{1}\Bigr)\Bigl(I_{a+n} \otimes \sigma_x\Bigr) \Bigl(\ket{0^a} \otimes I_{n+1} \Bigr)
    } \notag \\
    &=\bigCnorm{
      \bar{A} - \alpha \Bigl(\bra{0^a} \otimes I_{n+1} \Bigr) (U \otimes \gaiseki{0}{1} + U^\dag \otimes \gaiseki{1}{0}) \Bigl(\ket{0^a} \otimes I_{n+1} \Bigr)
    } \notag \\
    &=\bigCnorm{
      \bar{A} - \alpha
      \Bigl(
      (\bra{0^a} \otimes I_n)U     (\ket{0^a} \otimes I_n) \otimes \gaiseki{0}{1} +
      (\bra{0^a} \otimes I_n)U^\dag(\ket{0^a} \otimes I_n) \otimes \gaiseki{1}{0}
      \Bigr)
    } \notag \\
    &=\bigCnorm{
      \biggr(A - \alpha (\bra{0^a} \otimes I_n)U(\ket{0^a} \otimes I_n) \biggl) \otimes \gaiseki{0}{1} +
      \biggr(A^\dag - \alpha ( \bra{0^a}\otimes I_n)U^\dag(\ket{0^a} \otimes I_n) \biggl) \otimes \gaiseki{1}{0}
    } \notag \\
    &\le 2\epsilon.
  \end{align}
\end{proof}

The block-encoding of the inverse $\bar{A}^{-1}$ is constructed using the block-encoding of the extended matrix and QSVT. 
We represent the block-encoding of the original inverse $A^{-1}$ by 
inserting $\sigma_x$ operation after the block-encoding of $\bar{A}^{-1}$. 
The corresponding circuit is shown in \Cref{qc:return-inv}.

\begin{prop}\label{prop:return-inv}
  If $\tilde{U}$ is a $(\beta, a, \epsilon)$-block-encoding of $\bar{A}^{-1} = A^{-1} \otimes \gaiseki{1}{0} + (A^\dag)^{-1} \otimes \gaiseki{0}{1}$, 
  then $(I_a \otimes \SWAP_1^{n+1}) (I_{a+n} \otimes \sigma_x)\tilde{U} (I_a \otimes \SWAP_1^{n+1})$ is a $(\beta, a+1, \epsilon)$-block-encoding of $A^{-1}$.   
\end{prop}

\begin{proof}
  We can see that $A^{-1} = (I_n \otimes \bra{1})\bar{A}^{-1} (I_n \otimes \ket{0}) $. 
  In addition, $\bra{0^a} \otimes I_n \otimes \bra{1} = (I_n \otimes \bra{1})(\bra{0^a} \otimes I_{n+1})$ 
  and $\ket{0^a} \otimes I_n \otimes \ket{0} = (\ket{0^a} \otimes I_{n+1})(I_n \otimes \ket{0})$ hold. 
  Thus, 
  \begin{align}
    & \bigCnorm{
      A^{-1} - \beta \bigl(\bra{0^a} \otimes \bra{0} \otimes I_n  \bigr) \bigl(I_a \otimes \SWAP_1^{n+1} \bigr) (I_{a+n} \otimes \sigma_x)\tilde{U} \bigl(I_a \otimes \SWAP_1^{n+1} \bigr) \bigl(\ket{0^a} \otimes \ket{0} \otimes I_n \bigr)
    } \notag \\
    & \bigCnorm{
      A^{-1} - \beta \bigl(\bra{0^a} \otimes I_n \otimes \bra{0} \bigr) (I_{a+n} \otimes \sigma_x)\tilde{U} \bigl(\ket{0^a} \otimes I_n \otimes \ket{0} \bigr)
    } \notag \\
    &=\bigCnorm{
      (I_n \otimes \bra{1})\bar{A}^{-1} (I_n \otimes \ket{0})
      -
      \beta \bigl(\bra{0^a} \otimes I_n \otimes \bra{1} \bigr) \tilde{U} \bigl(\ket{0^a} \otimes I_n \otimes \ket{0} \bigr)
    } \notag \\
    &=\bigCnorm{
      (I_n \otimes \bra{1})
      \biggl(
        \bar{A}^{-1} - \beta (\bra{0^a} \otimes I_{n+1}) \tilde{U} (\ket{0^a} \otimes I_{n+1})
      \biggr)
      (I_n \otimes \ket{0})
    } \notag \\ 
    &\le \epsilon. 
  \end{align}
\end{proof}

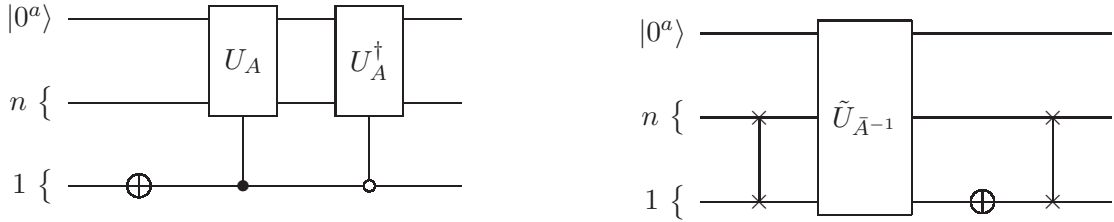
\begin{figure}[htbp]
  \begin{tabular}{cc}
    \begin{minipage}{0.45\hsize}
      \centering
      \[
        \Qcircuit @C=2em @R=2em {
        \lstick{\ket{0^a}} & \qw      & \multigate{1}{U_A} & \multigate{1}{U_A^\dag} & \qw \\
        \lstick{n \ \bigl\{} & \qw      & \ghost{U_A}     & \ghost{U_A^\dag}     & \qw \\
        \lstick{1 \ \bigl\{} & \targ & \ctrl{-1}     & \ctrlo{-1}          & \qw \\
        }
      \]
      \subcaption{
        \ $(\alpha, a, 2\epsilon)$-block-encoding of $n+1$-qubit matrix $\bar{A} = A \otimes \gaiseki{0}{1} + A^\dag \otimes \gaiseki{1}{0}$.
      }
      \label{qc:extendedbe}
    \end{minipage}
    \hspace{1mm}
    \begin{minipage}{0.5\hsize}
      \centering
        \[
          \Qcircuit @C=2em @R=2em {
          \lstick{\ket{0^a}} & \qw                      & \multigate{2}{\tilde{U}_{\bar{A}^{-1}}} & \qw      & \qw        & \qw \\
          \lstick{n \ \bigl\{} & \qswap                   & \ghost{\tilde{U}_{\bar{A}^{-1}}}        & \qw      & \qswap     & \qw \\
          \lstick{1 \ \bigl\{} & \qswap\qwx               & \ghost{\tilde{U}_{\bar{A}^{-1}}}        & \targ & \qswap\qwx & \qw \\
          }
        \]
        \subcaption{
          $(\beta, a+1, \epsilon)$-block-encoding of $n$-qubit matrix $A^{-1}$
        }
        \label{qc:return-inv}
    \end{minipage}
  \end{tabular}
  \caption{
    Block-encodings for $\bar{A}$ and $A^{-1}$,
    where $U_A$ is an $(\alpha, a, \epsilon)$-block-encoding of $n$-qubit matrix $A$ and
    $\tilde{U}_{\bar{A}^{-1} }$ is a $(\beta, a, \epsilon)$-block-encoding of $n+1$-qubit matrix $\bar{A}^{-1}$. 
    The circuits (a) and (b) is related to \Cref{prop:extendedbe,prop:return-inv}, respectively. 
  }
\end{figure}

\subsection{Diagonal matrix}
We now describe a naive construction of the block-encoding of diagonal matrices. 
This is used to construct the block-encoding of the block-diagonal matrix in the later section.   
The circuit of the block-encoding of diagonal matrices is shown in \Cref{qc:BEdiag}.

\begin{prop}[Block-encoding of diagonal matrix]\label{lem:diagbe}
  Let $D$ be an $m$-qubit diagonal matrix $D = \diag(d_0, d_1, \dots, d_{M-1}) $ and let $d_{\max} = \max_k \abs{d_k}$.  
  Suppose that $\cos(\phi_k/2) = \abs{d_k}/d_{\max}$ and $\varphi_k/2 = \arg(d_k)$.
  Then, $U_D = \sum_k (R_z(\varphi_k)R_y(\phi_k)) \otimes \gaiseki{k}{k}$ is 
  a $(d_{\max}, 1, 0)$-block-encoding of the diagonal matrix $D$, 
  where $R_y(\phi) = \e^{\im(\phi/2)\sigma_y}$ and $R_z(\varphi) = \e^{\im(\varphi/2)\sigma_z}$, respectively. 
\end{prop}

\begin{proof}
  By definitions, $\bra{0}(R_z(\varphi_k)R_y(\phi_k))\ket{0} = \frac{d_k}{d_{\max}}$. 
  Thus, from the definition of the block-encoding, it follows that 
  \begin{align}
    \bigCnorm{D - d_{\max} (\bra{0} \otimes I_m)U_D(\ket{0} \otimes I_m)}
    &=
    \bigCnorm{D - d_{\max} (\bra{0} \otimes I_m) \sum_k (R_z(\varphi_k)R_y(\phi_k)) \otimes \gaiseki{k}{k} (\ket{0} \otimes I_m)} \notag \\
    &=
    \bigCnorm{D - d_{\max} \sum_k \bra{0} R_z(\varphi_k)R_y(\phi_k)\ket{0} \otimes \gaiseki{k}{k} } \notag \\
    &= 0. 
  \end{align}
\end{proof}

\begin{figure}[htbp]
  \begin{center}
  \[
    \Qcircuit @C=1em @R=1em {
    \lstick{\ket{0}} & \gate{R_z(\varphi_0)} & \gate{R_y(\phi_0)} & \gate{R_z(\varphi_1)} & \gate{R_y(\phi_1)} & \qw & \cdots & & \gate{R_z(\varphi_{M-1})} & \gate{R_z(\phi_{M-1})} & \qw \\
    \lstick{}        & \ctrl{-1}             & \ctrl{-1}             & \ctrl{-1}             & \ctrl{-1}             & \qw & \cdots & & \ctrlo{-1}                & \ctrlo{-1}                & \qw \\
    \lstick{}        & \ctrl{-1}             & \ctrl{-1}             & \ctrl{-1}             & \ctrl{-1}             & \qw & \cdots & & \ctrlo{-1}                & \ctrlo{-1}                & \qw \\
    \lstick{ {\raisebox{2.4em}{$m$ \ } } }        & \ctrl{-1}             & \ctrl{-1}             & \ctrlo{-1}            & \ctrlo{-1}            & \qw & \cdots & & \ctrlo{-1}                & \ctrlo{-1}                & \qw     \gategroup{2}{1}{4}{1}{.8em}{\{}
    }
  \]   
  \end{center}   
\caption{
  $(d_{\max}, 1, 0)$-block-encoding $U_D = \sum_k (R_z(\varphi_k)R_y(\phi_k)) \otimes \gaiseki{k}{k}$ of the diagonal matrix $D = \diag(d_0, d_1, \dots, d_{M-1})$, 
  where $d_{\max} = \max_k \abs{d_k}$, $\phi_k = 2\arccos\left(\abs{d_k}/d_{\max} \right)$, and $\varphi_k = 2\arg(d_k)$. 
  This circuit is related to \Cref{lem:diagbe}. 
}
\label{qc:BEdiag}
\end{figure}
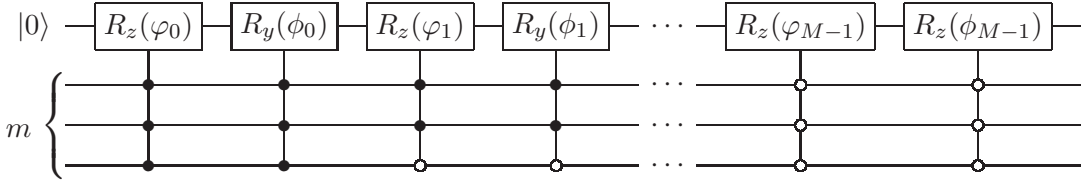

\section{Block-Encoding of Linear Combination of Inverse Matrices}\label{sec:MainAlgorithm}
In this section, we describe the construction of the block-encoding of a linear combination of inverses. 
Specifically, we construct the block-encoding of a matrix $F$ of the form 
\begin{align}
  F =  \sum_{k=0}^{M-1} w_k A_k^{-1}, 
  \label{eq:F}
\end{align}
under the assumption that the block-encoding $U_\sfA$ of a block-diagonal matrix $\sfA = \diag(A_0, A_1, \dots, A_{M-1})$ and 
the state-preparation pair $(P_L, P_R)$ for $\vec{w} = (w_0,w_1,\dots,w_{M-1})$ are given. 
The reason for considering this is that the matrix function, 
which is represented as a contour integral, can be approximated by 
a matrix in the form of \eqref{eq:F} using an appropriate quadrature. 
We use here the proposition on the inverse in the previou section. 
In general, the matrices $A_k$ obtained from the approximation of the matrix functions is not Hermitian. 
Therefore, we consider including the extention of matrices. 
By using the derivations in this section, we can replace the problem of constructing the block-encoding of the approximation 
with the problem of constructing a block-encoding of $\sfA$ and the state-preparation pair for $\vec{w}$.

%

In the following lemma, we describe how we can construct a block-encoding of a linear combination of the form $F = \sum_k w_k A_k^{-1}$ 
by combining the block-encoding of $\sfA^{-1}$ and the state-preparation pair for $\vec{w} = (w_0,w_1,\dots,w_{M-1})$. 

\begin{lemma}\label{thm:be-lc-invs}
  Suppose that $A_0,A_1,\dots,A_{M-1}$ are $n$-qubit matrices.  
  Let $\sfA = \diag(A_0, A_1, \dots, A_{M-1})$ be an $n+m$-qubit block-diagonal matrix and let $\beta$ be a real number such that $\norm{\sfA^{-1}} \le \beta $. 
  If $U_{\sfA^{-1}}$ is a $(\tilde{\beta}, \tilde{a}, \tilde{\epsilon} )$-block-encoding of an $n+m$-qubit matrix $\sfA^{-1}$ 
  and $(P_L, P_R)$ is a $(\mu, m, \delta)$-state-preparation pair for $\vec{w} = (w_0,w_1,\dots,w_{M-1})$, 
  then 
  \begin{align}
   U_F := (I_{\tilde{a}} \otimes P_L^\dag \otimes I_n)U_{\sfA^{-1}}(I_{\tilde{a}} \otimes P_R \otimes I_n)
  \end{align}
  is a $(\tilde{\beta}\mu, \tilde{a}+m, \beta \delta  + \mu\tilde{\epsilon} )$-block-encoding of $F = \sum_{j=0}^{M-1} w_j A_j^{-1}$. 
\end{lemma}
\begin{proof}
  For any $m$-qubit matrix $P$, we have 
  \begin{align}
    \Bigl( \bra{0^{\tilde{a}}} \otimes \bra{0^m} \otimes I_n \Bigr) \Bigl( I_{\tilde{a}} \otimes P \otimes I_n \Bigr)
    &= \Bigl( \bra{0^m} \otimes I_n \Bigr) \Bigl( \bra{0^{\tilde{a}}} \otimes I_{n+m} \Bigr) \Bigl( I_{\tilde{a}} \otimes P \otimes I_n \Bigr) \notag \\
    &= \Bigl( \bra{0^m}  \otimes I_n \Bigr) \Bigl(P \otimes I_n\Bigr) \Bigl( \bra{0^{\tilde{a}}} \otimes I_{n+m} \Bigr) \notag \\
    &= \Bigl( \bra{0^m}P \otimes I_n ) \Bigl( \bra{0^{\tilde{a}}} \otimes I_{n+m} \Bigr). 
  \end{align}
  We write $P_L\ket{0^m} = \ket{c} = \sum_{k=0}^{M-1}c_k\ket{k}, P_R\ket{0^m} = \ket{d} = \sum_{k=0}^{M-1}d_k\ket{k}$. 
  From the above, 
  \begin{align*}
    &  \Bigl( \bra{0^{a+m}} \otimes I_n \Bigr) U_F \Bigl( \ket{0^{a+m}}  \otimes I_n \Bigr)
       \\
    &= \Bigl( \bra{0^{\tilde{a}}} \otimes \bra{0^m} \otimes I_n \Bigr)
       (I_{\tilde{a}} \otimes P_L^\dag \otimes I_n)
       U_{\sfA^{-1}}
       (I_{\tilde{a}} \otimes P_R \otimes I_n)
       \Bigl( \ket{0^{\tilde{a}}} \otimes \ket{0^m} \otimes I_n \Bigr)
       \\
    &= \Bigl( (\bra{0^m}P_L^\dag ) \otimes I_n \Bigr)
       \Bigl(  \bra{0^{\tilde{a}}} \otimes I_{n+m} \Bigr)
       U_{\sfA^{-1}}
       \Bigl( \ket{0^{\tilde{a}}} \otimes I_{n+m} \Bigr)
       \Bigl( (P_R\ket{0^m}) \otimes I_n \Bigr)
      \\
    &= \Bigl(  \bra{c} \otimes I_n \Bigr) \Bigl( \bra{0^{\tilde{a}}} \otimes I_{n+m} \Bigr) U_{\sfA^{-1}} \Bigl( \ket{0^{\tilde{a}}} \otimes I_{n+m} \Bigr) \Bigl( \ket{d} \otimes I_n \Bigr) 
  \end{align*}
  holds. 
  By the triangle inequality, we have 
  \begin{align}
    & \bigDnorm{ F - \tilde{\beta}\mu \Bigl( \bra{c} \otimes I_n \Bigr) \Bigl( \bra{0^{\tilde{a}}} \otimes I_{n+m} \Bigr) U_{\sfA^{-1}} \Bigl( \ket{0^{\tilde{a}}} \otimes I_{n+m} \Bigr) \Bigl( \ket{d} \otimes I_n \Bigr) }
    \notag \\
    &\le \bigDnorm{
      F - \mu \Bigl( \bra{c} \otimes I_n \Bigr) \sfA^{-1} \Bigl( \ket{d} \otimes I_n \Bigr)
    }
    \notag \\
    & \qquad\quad + \mu \Biggr\lVert
      \Bigl( \bra{c}  \otimes I_n \Bigr) \sfA^{-1} \Bigl( \ket{d} \otimes I_n \Bigr) -
      \tilde{\beta}\Bigl( \bra{c} \otimes I_n \Bigr) \Bigl( \bra{0^{\tilde{a}}} \otimes I_{n+m} \Bigr) U_{\sfA^{-1}} \Bigl( \ket{0^{\tilde{a}}} \otimes I_{n+m} \Bigr) \Bigl( \ket{d} \otimes I_n \Bigr)
    \Biggr\rVert
    \notag \\
    & = \bigDnorm{
      F - \mu \Bigl( \bra{c} \otimes I_n \Bigr) \sfA^{-1} \Bigl( \ket{d} \otimes I_n \Bigr)
    }
    \label{eq:at-thm:be-lc-invs-1} \\
    & \qquad\quad + \mu \Biggr\lVert
      \Bigl( \bra{c}  \otimes I_n \Bigr)
      \Bigl\{ \sfA^{-1} - \tilde{\beta} \Bigl( \bra{0^{\tilde{a}}} \otimes I_{n+m} \Bigr) U_{\sfA^{-1}} \Bigl( \ket{0^{\tilde{a}}} \otimes I_{n+m} \Bigr)
      \Bigr\}
      \Bigl( \ket{d} \otimes I_n \Bigr)
    \Biggr\rVert. 
    \label{eq:at-thm:be-lc-invs-2}
  \end{align}
  We focus on the first part of the above, \eqref{eq:at-thm:be-lc-invs-1}. 
  As $\sfA^{-1} = \sum_{j=0}^{M-1} \gaiseki{j}{j} \otimes A_k^{-1}$, we have that    
  \begin{align*}
   \bigDnorm{
      F - \mu \Bigl( \bra{c} \otimes I_n \Bigr) \sfA^{-1} \Bigl( \ket{d} \otimes I_n \Bigr)
    }
  &=\bigDnorm{
      F - \mu \Bigl( \bra{c} \otimes I_n \Bigr)
      \l( \sum_{k=0}^{M-1} \gaiseki{k}{k}\otimes A_k^{-1} \r)
      \Bigl( \ket{d} \otimes I_n \Bigr)
    } \\
    &=
    \bigDnorm{
        \sum_{k=0}^{M-1}w_kA_k^{-1}  - \mu  \sum_{k=0}^{M-1} c_k^\ast d_k A_k^{-1}
      } \\
    &=
    \bigDnorm{
        \sum_{k=0}^{M-1} (w_k - \mu c_k^\ast d_k) A_k^{-1}
      } \\
    &\le \sum_{k=0}^{M-1} \abs{w_k - \mu c_k^\ast d_k} \norm{A_k^{-1}} \\
    &\le \delta\beta. 
  \end{align*}
  The second part is bounded by $\mu\tilde{\epsilon}$ because $U_{\sfA^{-1}}$ is the $(\tilde{\beta}, \tilde{a}, \tilde{\epsilon} )$-block-encoding of $\sfA^{-1}$.    
\end{proof}

Using the above lemma, 
we now describe the block-encoding of the linear combination of inverses \eqref{eq:F} in the case where the block-encoding of $\sfA$ and the state-preparation pair for $\vec{w}$ are given. 
This is used as a component in the later section to construct the block-encoding of matrix functions. 

\begin{prop}[Block-encoding of linear combination of inverses]\label{thm:be-F}
  Suppose that $A_0,A_1,\dots A_{M-1}$ are $n$-qubit matrices.  
  Let $\alpha$ and $\beta$ be real numbers such that $\max_k \{\norm{A_k}\} \le \alpha$ and $\max_k \{\norm{A_k^{-1}}\} \le \beta$. 
  Let $r > 0$ and $F = \sum_k w_k A_k^{-1}$. 
  If $U_\sfA$ is an $(\alpha, a, \epsilon)$-block-encoding of $\sfA = \diag(A_0,A_1,\dots,A_{M-1})$ 
  and $(P_L, P_R)$ is a $(\mu, m, \delta_{\text{sp}})$-state-preparation pair for $\vec{w} = (w_0,w_1,\dots,w_{M-1})$, 
  then there is a circuit $U_F$ that is the $(\tau, a+m+3, \eta) $-block-encoding of $rF$, where 
  \begin{align}
  \tau = r\tilde{\beta}\mu = \frac{16}{3}r\beta\mu, \qquad
  \eta = r(\beta\delta_{\text{sp}} + \mu\tilde{\epsilon})
  = r\beta\left(\delta_{\text{sp}} + \frac{16}{3} \mu \left(4d \sqrt{\frac{2\epsilon}{\alpha}} + \delta\right)\right), 
  \end{align}
  and $\tilde{\beta} = \frac{16}{3}\beta, \tilde{\epsilon} = \left(4d\sqrt{ \frac{2\epsilon}{\alpha} } + \delta \right) \tilde{\beta} \ (\delta \in (0, \frac{3}{4}])$. 
  Further, the circuit $U_F$ consists of $O(d)$ uses of $U_\sfA$, single uses of $P_L,P_R$, and $O((a+1)d)$ other one- and two-qubit gates, 
  where   
  \begin{align}
    d = O\left(\alpha\beta \log\left(\frac{1}{\delta}\right)\right). 
  \end{align}
  Moreover, we can compute a description of such a circuit with a classical computer in $\poly(d, \log(\frac{1}{\delta}))$ time.   
\end{prop}
\begin{proof}
  Let $\bar{\sfA}$ be an extended matrix of $\sfA$, i.e., let $\bar{\sfA} = \sfA \otimes \gaiseki{0}{1} + \sfA^\dag \otimes \gaiseki{1}{0}$. 
  From \Cref{prop:extendedbe}, we can construct an $(\alpha, a, 2\epsilon)$-block-encoding $U_{\bar{\sfA}}$ of 
  the extended matrix $\bar{\sfA}$ through two uses of the operation $U_\sfA$. 
  The extended matrix $\bar{\sfA}$ is Hermitian and satisfies $\norm{\bar{\sfA}} = \norm{\sfA}$. 
  The eigenvalues of $\bar{\sfA}$ range over $\calI = [-\alpha, -1/\beta] \cup [1/\beta, \alpha]$. 
  Thus, from \Cref{prop:inversebe}, 
  a $(\tilde{\beta}, a+2, \tilde{\epsilon} )$-block-encoding $\tilde{U}_{\bar{\sfA}^{-1}}$ of $\bar{\sfA}^{-1}$ 
  can be constructed by $O(d)$ uses of $U_\sfA$ and $O((a+1)d)$ other one- and two-qubit gates. 
  Moreover, we can compute a description of such a circuit with a classical computer in $\poly(d, \log(\frac{1}{\delta}))$ time.  
  Hence, from  \cref{prop:return-inv}, 
  \begin{align}
    \tilde{U}_{\sfA^{-1}} := 
    \Bigl(I_{a+2} \otimes \SWAP_1^{n+m+1} \Bigr)
    (I_{a+2+n+m} \otimes \sigma_x)\tilde{U}_{\bar{\sfA}^{-1}}
    \Bigl(I_{a+2} \otimes \SWAP_1^{n+m+1} \Bigr)
  \end{align}
  is the 
  $(\tilde{\beta}, a+3, \tilde{\epsilon} )$-block-encoding of $\sfA^{-1}$. 
  Thus, 
  it follows from \cref{thm:be-lc-invs} that 
  \begin{align}
   U_F = (I_{a+3} \otimes P_L^\dag \otimes I_n ) \tilde{U}_{\sfA^{-1}} (I_{a+3} \otimes P_R \otimes I_n ) 
  \end{align}
  is a $(\tilde{\beta}\mu, a+m+3, \beta\delta_{\text{sp}} + \mu\tilde{\epsilon})$-block-encoding of $\sum_k w_k A_k^{-1}$. 
  In other words, 
  $U_F$ is the $(r\tilde{\beta}\mu, a+m+3, r(\beta\delta_{\text{sp}} + \mu\tilde{\epsilon}))$-block-encoding of $r\sum_k w_k A_k^{-1}$. 
\end{proof}

\section{Block-Encoding of Matrix Functions}\label{sec:MatrixFunction}
In the previous section, we proved that we can construct the block-encoding of the linear combination of the inverses $F = \sum_k w_k A_k^{-1}$
when the block-encoding of $\sfA = \diag(A_0,A_1,\cdots,A_{M-1})$ and the state-preparation pair for $\vec{w} = (w_0,w_1,\dots,w_{M-1})$ are given. 
In this section, we consider the block-diagonal matrix $\sfA$ and the vector $\vec{w}$ with respect to the coefficients from the approximation for matrix functions.  
We then construct the block-encoding for matrices and the state-preparation pair for vectors. 
We consider two cases. 
\Cref{subsec:matfunc-circle} examines the case in which the contour is a circle, 
and  \Cref{subsec:notcircle} investigates the case in which the contour is not a circle.  
In particular, for the case in which the contour is a circle, we also consider the bounds of the singular values of $\sfA$.

\subsection{Cases in which the contour is a circle}\label{subsec:matfunc-circle}

In this section, we consider a case in which a circular contour encloses all eigenvalues of the matrix. 
In such the case, the matrix function of the form \Cref{eq:deffA} can be approximated by the trapezoidal rule with high accuracy. 
In the following definition, we describe the approximation. For the derivation of the approximation by the trapezoidal rule, see the appendix. 
\begin{defi}\label{def:trapezoidal}
  Let $M = 2^m \ (m \in \N)$ and let $\theta_k = \frac{2\pi\im}{M}k \ ( k \in [M])$. 
  Suppose that $f$ is an analytical function on the disk $\abs{z - z_0} \le R$
  and suppose that an $n$-qubit matrix $A$ satisfies $\norm{A - z_0 I_n} < r < R$. 
  Then, we define an approximation $f_M(A)$ to $f(A)$ by $M$-point trapezoidal rule as 
  \begin{align}
    f_M(A) = \frac{1}{M} \sum_{k=0}^{M-1} f(z_0 + r\e^{\im\theta_k}) r\e^{\im\theta_k} ( (z_0 + r\e^{\im\theta_k})I_n - A )^{-1}. 
    \label{eq:f_MA}
  \end{align}    
  Putting $w_k = \frac{1}{M}f(z_0 + r\e^{\im\theta_k})\e^{\im\theta_k}$ and $A_k = (z_0 + r\e^{\im\theta_k})I_n - A$, 
  we write the above as $f_M(A) =r \sum_{k=0}^{M-1} w_k A_k^{-1}$. 
\end{defi}

In the following, we first consider the block-encoding of the block-diagonal matrix $\sfA$. 
As mentioned in the earlier sections, we decompose $\sfA$ to the linear combination of tensor products and then construct the block-encoding of it using \Cref{prop:linearcobination_tensorproduct}. 
Next, we derive the upper bounds of $\norm{\sfA}, \norm{\sfA^{-1}}$. 
Further, we consider a state-preparation-pair for a vector that approximates $\vec{w} = (w_0,w_1,\dots,w_{M-1})$. 

\begin{prop} 
  \label{lem:sfA-case-circle}
  We use the same notation as in \Cref{def:trapezoidal}. 
  Suppose that $U_A$ is an $(\alpha, a, \epsilon_A)$-block-encoding of an $n$-qubit matrix $A$. 
  Let  $(Q_L, Q_R)$ be an $(\alpha', 2, 0)$-state-preparation pair for $(z_0, r, -\alpha, 0)$, 
  where $\alpha' = r + \alpha + \abs{z_0}$. 
  Let $\SWAP_{a, m} = \prod_{i=1}^a \SWAP_{i+m}^i$ and 
  \begin{align}
    W = 
      \gaiseki{00}{00} \otimes I_{a+n+m}
      + \gaiseki{01}{01} \otimes \mathsf{R} \otimes I_{a+n} 
      + \gaiseki{10}{10} \otimes I_m \otimes U_A
      + \gaiseki{11}{11} \otimes I_{a+n+m}, 
  \end{align}
  where $\mathsf{R} := \diag(\e^{\im\theta_0}, \e^{\im\theta_1},\dots,\e^{\im\theta_{M-1}}) = R_0 \otimes R_1 \otimes \cdots \otimes R_{m-1}$ and $R_j := \diag(1, \e^{2\pi\im/2^j})$. 
  Then,  
  \begin{align}
    U_\sfA &= \Bigl( Q_L^\dag \otimes \SWAP_{a, m} \otimes I_n \Bigr) W \Bigl( Q_R \otimes \SWAP_{a, m} \otimes I_n  \Bigr)
  \end{align}
  is an $(\alpha', a+2, \epsilon_A)$-block-encoding of the block-diagonal matrix $\sfA = \diag(A_0,A_1,\dots,A_{M-1})$. 
\end{prop}
\begin{proof}
  We can see that the block-diagonal matrix $\sfA$ can be decomposed as 
  $\sfA = z_0 I_{m+n} + r \sfR \otimes I_n - I_m \otimes A$. 
  Thus, the proposition follows from \Cref{prop:linearcobination_tensorproduct}. 
\end{proof}

\begin{prop}[Upper and lower bound of the singular values of $\sfA$]\label{lem:sfA-case-circle-cond}
  We use the same notation as in \Cref{def:trapezoidal}. 
  Let $\sfA = \diag(A_0,A_1,\dots,A_{M-1})$. Then,
  \begin{align}
    \norm{\sfA} \le r + \abs{z_0} + \norm{A} \quad\text{and}\quad \norm{\sfA^{-1}} \le \frac{1}{r - \norm{A - z_0I_n}}. 
  \end{align}
\end{prop}
\begin{proof}
  The matrix $\sfA$ can be decomposed as $\sfA = r\diag(\e^{\im\theta_0},\dots,\e^{\im\theta_{M-1}})\otimes I_n + I_m \otimes (z_0 I_n - A)$. 
  Thus, $\norm{\sfA} \le r + \norm{z_0 I -  A} \le r + \abs{z_0} + \norm{A}$ holds. 
  Next, we consider the upper bound of $ \norm{\sfA^{-1}}$. 
  We can see that $\norm{\sfA^{-1}} = \max_k \{ \norm{A_k}^{-1} \}$ because $\sfA$ is the block-diagonal matrix in which each diagonal block is $A_k$. 
  From $A_k = (z_0 + r\e^{\im\theta_k})I_n - A = r\e^{\im\theta_k}I_n - (-z_0I_n + A)$ and $r > \norm{-z_0 I - A}$, 
  $\norm{A_k^{-1}} \le \frac{1}{r}\sum_{j=0}^\infty \norm{((-z_0I_n + A)/r\e^{\im\theta_k})} \le \frac{1}{r}(1 - \frac{\norm{-z_0I_n + A}}{r})^{-1} = 
  (r - \norm{z_0I - A})^{-1}$
  holds according to the Neumann series.  
\end{proof}

\begin{prop}\label{lem:st-case-circle}
  We use the same notation as in \Cref{def:trapezoidal}. 
  Let $f(z) = \sum_{\ell=0}^\infty a_\ell (z - z_0)^\ell$ be the Taylor series of $f(z)$ and let $\tilde{f}_L(z) = \sum_{\ell=0}^{L-1} a_\ell (z - z_0)^\ell$ be the truncated series at $L$-th term. 
  Let $\delta_L = (1 - \frac{r}{R})^{-1}(r/R)^L$ and let $\tilde{w}_k = \frac{1}{M}\tilde{f}_L(z_0 + r\e^{\im\theta_k})\e^{\im\theta_k}, \ \vec{\tilde{w}} = (\tilde{w}_0, \tilde{w}_1, \dots, \tilde{w}_{M-1})$. 
  We define $P_L$ and $P_R$ such that 
  \begin{align}  
    P_L\ket{0^{\otimes m}} = \frac{1}{ \sqrt{ \norm{ \vec{\tilde{w} } }_1 } }\sum_{k=0}^{M-1}\sqrt{\tilde{w}_k^\ast}\ket{k}
    \quad \text{and} \quad
    P_R\ket{0^{\otimes m}} = \frac{1}{ \sqrt{ \norm{ \vec{\tilde{w} } }_1 }}\sum_{k=0}^{M-1}\sqrt{\tilde{w}_k}\ket{k},
  \end{align}
  respectively. Then, $(P_L, P_R)$ is a $(\norm{\vec{w}}_1, m, 2 \norm{f}_{\infty} \delta_L)$-state-preparation-pair for $\vec{w}$, 
  where $\norm{f}_{\infty} = \max_{z:\abs{z - z_0} \le R} \abs{f(z)} $. 
  Moreover, we can perform $P_L, P_R$ with $O(M)$ runtime, and we can compute a description of such a circuit with a classical computer in $O(ML)$.     
\end{prop}
\begin{proof}
To obtain $\tilde{w}_k$, we need $O(L)$ runtime. Thus, we can obtain $\vec{\tilde{w}}$ with gate complexity of $O(ML)$. 
The error in the $\tilde{w}_k$ is bounded as follows: 
\begin{align}
  \abs{w_k - \tilde{w}_k}
  = \left\lvert \frac{\e^{\im\theta_k}}{M} \sum_{\ell = L}^\infty a_\ell (r\e^{\im\theta_k})^\ell \right\rvert 
  \le \frac{1}{M} \sum_{\ell = L}^\infty \frac{\norm{f}_{\infty}}{R^\ell} r^\ell 
  = \frac{\norm{f}_{\infty}}{M}\frac{1}{1 - \frac{r}{R}} \left(\frac{r}{R}\right)^L = \frac{\norm{f}_{\infty}}{M}\delta_L,
\end{align}
where we used Cauchy's estimate: $\abs{a_\ell} \le \norm{f}_\infty/R^\ell$. 
Thus, we have $\norm{\vec{w} - \vec{\tilde{w}}}_1 \le \sum_{k=0}^{M-1} \abs{w_k - \tilde{w}_k} \le \norm{f}_{\infty}\delta_L$. 
Therefore, 
\begin{align}
  \sum_{k=0}^{M-1} \abs{ \norm{\vec{w}}_1 \frac{\tilde{w}_k}{\norm{\vec{\tilde{w}}}_1} - w_k }
  = 
  \norm{\vec{w}}_1 \bigCnorm{\frac{\vec{\tilde{w}}}{ \norm{\vec{\tilde{w}}}_1 } - \frac{\vec{w}}{\norm{\vec{w}}_1} }_1
  \le 
  2 \norm{\vec{\tilde{w}} - \vec{w}}_1
  \le 
  2 \norm{f}_{\infty}\delta_L. 
\end{align}
Hence, $(P_L, P_R)$ is the $(\norm{\vec{w}}, m, 2\norm{f}_{\infty}\delta_L)$-state-preparation-pair for 
$\vec{w}$. 
By using the state preparation method described in \cite{SBM2006}, we can perform $P_L$ and $P_R$ with $O(M)$ runtime. 
\end{proof}

Combining the aboves and \Cref{thm:be-F}, we can derive the following to describe the construction of the block-encoding of the approximation $f_M(A)$. 
From this theorem, we can construct the block-encoding of matrix functions. 

\begin{thm}[Block-encoding of $f_M(A)$]\label{Thm:MainResult:MatrixFunction1}
  We use the same notation as in \Cref{def:trapezoidal}. 
  Let $\delta > 0$.
  Let $\norm{f}_{\infty} = \max_{z:\abs{z - z_0} \le R} \abs{f(z)} $ and let $\delta_L = (1 - \frac{r}{R})^{-1}(r/R)^L$. 
  Suppose that $U_A$ is an $(\alpha, a, \epsilon_A)$-block-encoding of an $n$-qubit matrix $A \in \C^{N \times N}$. 
  Then, we can construct a quantum circuit $U_{f_M(A)}$ 
  that is a $ (\tau, a+m+5, \eta  )$-block-encoding of $f_M(A)$, where 
  \begin{align}\label{eq:factorerror_fm}
    \tau = \frac{16}{3}\frac{\norm{\vec{w}}}{1 - \frac{\norm{A - z_0 I_n}}{r} }, \quad
    \eta = 
    \frac{\norm{f}_\infty}{1 - \frac{\norm{A - z_0 I_n}}{r} }
    \left(
    2\delta_L + \frac{16}{3} \left(4d\sqrt{\frac{2\epsilon_A}{r + \abs{z_0} + \alpha }} +  \delta\right)
    \right), 
  \end{align}
  Further, the circuit $U_{f_M(A)}$ consists of $O(d)$
  uses of the $U_A$ operation and $O((a+ \log(M) + 1)d + M)$ other one- and two-qubit gates, 
  where
  \begin{align}
    d = O\left( \frac{r + \alpha + \abs{z_0}}{r - \norm{A - z_0 I}} \log \frac{1}{\delta}  \right). 
  \end{align}
  Moreover, we can compute a description of such a circuit with a classical computer in $\poly(L, M, d, \log\frac{1}{\delta})$ time.   
\end{thm}
\begin{proof}
  From  \cref{lem:st-case-circle}, 
  the $(\norm{\vec{w}}, m, 2\norm{f}_{\infty}\delta_L )$-state-preparation pair $(P_L, P_R)$ for 
  the vector $\vec{w} = (w_0, w_1,\dots, w_{M-1})$
  can be implemented using a circuit consisting of $O(M)$ one- and two-qubit gates. 
  Moreover, such a circuit can be described in $O(ML)$ runtime on a classical computer. 
  Let us consider a block-diagonal matrix $\sfA = (A_0,A_1, \dots, A_{M-1})$. 
  We can see from \cref{lem:sfA-case-circle-cond} that 
  $\sfA$ satisfies $\norm{\sfA} \le r + \abs{z_0} + \norm{A} \le r + \abs{z_0} + \alpha =: \alpha'$ and $\norm{\sfA^{-1}} \le (r - \norm{A - z_0 I_n })^{-1} =: \beta'$. 
  Moreover, by \cref{lem:sfA-case-circle}, 
  the $(\alpha', a+2, \epsilon_A)$-block-encoding $U_\sfA$  of $\sfA$ can be constructed  
  using one $U_A$ operation and $O(m + a) = O(\log M + a)$ other one- and two-qubit gates. 
  Thus, the theorem follows from \Cref{thm:be-F}.  
\end{proof}


The error between $f(A)$ and $f_M(A)$ is upper bounded as $  \bigAnorm{f(A) - f_M(A)} \le \epsilon_M$ \cite[Theorem 18.1]{TW14}, where
\begin{align}
  \epsilon_M = 
  \frac{\norm{f}_\infty}{1 - \frac{\norm{-z_0 I + A}}{R}} 
  \left( \frac{1}{1 - \left(\frac{\norm{-z_0 I + A}}{r}\right)^M} \left(\frac{\norm{-z_0 I + A}}{r}\right)^M
  +
  \frac{1}{1 - \left(\frac{r}{R}\right)^M} \left(\frac{r}{R}\right)^M \right), 
        \label{eq:erroroffMA}
\end{align}
The details of the derivation of the upper bound is described in the appendix. 
Thus, the $(\tau, a+m+5, \eta)$-block-encoding in \Cref{Thm:MainResult:MatrixFunction1} is also 
the $(\tau, a+m+5, \eta + \epsilon_M)$-block-encoding of $f(A)$.

Compared with the technique described in \cite{TOSU2020}, 
this method has the novelty that the algorithm is described in terms of the block-encoding framework. 
Additionally, this method can be applied to cases in which the circle is not centered on the origin. 
However, the overall complexity has increased from $O(\log M)$ to $O(M)$ 
with respect to the number of integration points $M$, but the order of $M$ is the logarithm of the inverse of the desired error. Therefore, this does not have a strongly negative effect. 

\subsection{Cases in which the contour is not a circle}\label{subsec:notcircle}
The approximation of the matrix functions in cases where the contour is not a circle is also be considered. 
In such cases, the parameters that appear in the approximation by the quadratures become  intricate. 
Therefore, unlike the case of a circle, it is difficult to use the structure of the parameters.  
In this subsection, we consider a block-encoding of such an approximation. Specifically, we consider the block-encoding of a matrix defined by the following. 
\begin{defi}\label{def:calFMA}
    Let $w_k, z_k, x_k \in \C \ (k = 0,1,\dots,M-1)$ and let $r > 0$. 
    Then, we define 
    \begin{align}
      \mathcal{F}_M(A) = r\sum_{k=0}^{M-1}w_k (y_k I + z_k A)^{-1} 
      \label{eq:calFM}
    \end{align}
  Putting $A_k = y_kI + z_k A$, we write the above as $\mathcal{F}_M(A) = r \sum_{k=0}^{M-1} w_k A_k^{-1}$. 
\end{defi}
As the approximations have this form, it is sufficient to consider them in this manner. 
Through the block-encoding formulation described here, 
it is possible to construct block-encodings for various matrix functions.

Here, we also consider the block-diagonal matrix $\sfA = \diag(A_0,A_1,\dots,A_{M-1})$ and its decomposition as a linear combination of the tensor products. 
Using \Cref{prop:linearcobination_tensorproduct,lem:diagbe}, we construct the block-encoding of $\sfA$. 


\begin{prop}\label{lem:sfA-case-random} 
We use the same notation as \Cref{def:calFMA}. 
Suppose that $U_A$ is an $(\alpha, a, \epsilon_A)$-block-encoding of an $n$-qubit matrix $A \in \C^{N \times N}$. 
Let $U_Y$ be a $(y_{\max}, 1, 0)$-block-encoding of the diagonal matrix $Y = \diag(y_0, \dots, y_{M-1})$
and $U_Z$ be a $(z_{\max}, 1, 0)$-block-encoding of the diagonal matrix $Z = \diag(z_0, \dots, z_{M-1})$, 
where $y_{\max} = \max_k \abs{y_k}$ and $z_{\max} = \max_k \abs{z_k}$. 
Let $(Q_L, Q_R)$ be an $(\alpha', 1, 0)$-state-preparation pair for $(y_{\max}, z_{\max}\alpha )$, 
where $\alpha' = y_{\max} + z_{\max}\alpha$. 
Let $\SWAP_{a, m} = \prod_{i=1}^a \SWAP_{m+i}^i$ and let 
\begin{align}
  W = \gaiseki{0}{0} \otimes U_Y \otimes I_{a+n} + 
  \gaiseki{1}{1} \otimes U_Z \otimes U_A. 
\end{align}
Then, 
\begin{align}
  U_\sfA :=
  \Bigl(Q_L^\dag \otimes I_1 \otimes \SWAP_{a, m} \otimes I_n \Bigr)
  W
  \Bigl(Q_R      \otimes I_1 \otimes \SWAP_{a, m} \otimes I_n \Bigr)
\end{align}
is an $(\alpha', a+2, z_{\max}\epsilon_A)$-block-encoding of $\sfA = \diag(A_0, A_1, \dots, A_{M-1})$, 
which can be constructed using a single $U_A$ operation and $O(M + a)$ one- and two-qubit gates. 
\end{prop}
\begin{proof}
  We can see that $\sfA$ is decomoposed as $\sfA = Y \otimes I_n + Z \otimes A$. 
  By \Cref{lem:diagbe}, we can implement $U_Y, U_Z$ with $O(M)$ one- and two-qubit gates.  
  Thus, the proposition follows from \Cref{prop:linearcobination_tensorproduct}. 
\end{proof}

  Combining the block-encoding of $\sfA$ and the method in \cite{SBM2006} to prepare the state with respect to the coefficients, 
  we can obtain the block-encoding of $\mathcal{F}_M(A)$ through the quantum algorithm for the linear combination of inverses (\Cref{thm:be-F}). 
  Here, we assume that the bounds for the singular values of $\sfA$ are given. 

\begin{thm}[Block-encoding of $\mathcal{F}_M(A)$]
  Let $\beta'$ be a real number such that $\max_k \{ \norm{(y_kI_n + z_k A )^{-1}} \} \le \beta'$. 
  Suppose that $U_A$ is an $(\alpha, a, \epsilon_A)$-block-encoding of an $n$-qubit matrix $A$ and 
  suppose that $(P_L, P_R)$ is an $(\mu, m, \delta_w)$-state-preparation pair for $\vec{w}$. 
  Then, we can construct a quantum circuit $U_{\mathcal{F}_M(A)}$ that is a 
  $ ( \tau, a+m+5, \eta)$-block-encoding of $\mathcal{F}_M(A)$, 
  where 
  \begin{align}\label{eq:factorerror_Fm}
    \tau =  \frac{16}{3} \cdot r\beta'\mu, 
    \quad 
    \eta = r \beta' \left( \delta_w + \frac{16}{3} \cdot \mu 
    \left(4d\sqrt{\frac{2 z_{\max} \epsilon_A}{y_{\max} + z_{\max}\alpha}  } + \delta\right) \right), 
  \end{align}
  Further, $U_{\mathcal{F}_M(A)}$ consists of $O(d)$ uses of the $U_A$ operation and $O((a+M+1)d+M)$ 
  other one- and two-qubit gates, where 
  \begin{align}
    d = O \left( ( y_{\max} + z_{\max}\alpha ) \beta' \log\frac{1}{\delta} \right). 
  \end{align}
  Moreover, 
  we can compute a description of such a circuit with a classical computer in $\poly(d, \log\frac{1}{\delta})$ time. 
\end{thm}
\begin{proof}
  The $(\mu, m, \delta_w)$-state-preparation pair $(P_L, P_R)$ for $\vec{w} = (w_0, w_1,\dots, w_{M-1})$
  can be implemented using a circuit consisting of $O(M)$ one- and two-qubit gates. 
  By \cref{lem:sfA-case-random}, 
  the $(\alpha', a+2, z_{\max} \epsilon_A)$-block-encoding $U_\sfA$ of $\sfA$ 
  consists of a single $U_A$ operation and $O(M + a)$ other one- and two-qubit gates. 
  Furthermore, the singular values of the block-diagonal matrix $\sfA$ range over $[1/\beta', \alpha']$, 
  where $\alpha' = y_{\max} + z_{\max}\alpha$. 
  Thus, the theorem follows from \Cref{thm:be-F}. 
\end{proof}

When $\norm{\mathcal{F}(A) - \mathcal{F}_M(A)} \le \epsilon_M$, 
the block-encoding of $\mathcal{F}_M(A)$ is also a $(\tau, a+m+5, \eta + \epsilon_M)$-block-encoding of $\mathcal{F}(A)$. 
The above results include the case in which the contour is a circle. 
Indeed, if we set $w_k,y_k$ and $z_k$ in \Cref{eq:calFM} to be consistent with \Cref{eq:f_MA} (i.e., 
$w_k = f(z_0 + r\e^{\im\theta_k})\e^{\im\theta_k}/M, y_k = (z_0 + r\e^{\im\theta_k})$ and $z_k = -1$),
then the normalization factor and error for \Cref{eq:factorerror_Fm} 
are the same as in \Cref{eq:factorerror_fm}. 
However, the number of one- and two-qubit gates increases.  

\section{Conclusion}\label{sec:conclusion}

In this paper, we considered the quantum algorithm to perform the linear combination of inverses. 
The motivation for this is that the matrix function based on the contour integrals can be approximated by the 
linear combination of the inverses of the shifted matrices. 
Utilizing this algorithm,
we proposed the quantum algorithm for the matrix functions approximated by the trapezoidal rule. 
Unlike our previous work  \cite{TOSU2020}, 
this algorithm can be applied to 
the case not only at the origin but also any circle on the complex plain. 

Moreover, we considered the quantum algorithm to perform the matrix of the form $\sum_k w_k(z_kI + y_kA)^{-1}$. 
This form arises in the approximations by the various quadratures, 
for example, trapezoidal rule, Gauss-Laguerre quadrature, and the double exponential formula \cite{TM74}. 
Moreover, this form arises not only contour integral representations, but also real integral representations. 
Therefore, the algorithms presented here can be applied to various cases. 
These quantum algorithms have been described concretely in terms of the block-encoding framework. 
Thus, the algorithms can be used as a subroutine easily.

The construction for a block-encoding of specific matrix functions (e.g., the matrix logarithm) 
and the analysis of the parameters in the block-encoding will be investigated in future work. 
Moreover, the investigation of the trace estimation of the matrix functions via the methods here is a topic for further research. 


\section*{Acknowledgements}
This work was supported by 
MEXT Quantum Leap Flagship Program (QLEAP) Grant Numbers JPMXS0118067394 and JPMXS0120319794. 
This work was supported by JSPS KAKENHI Grant Numbers JP18H05392 and JP20H00581.  
We thank Stuart Jenkinson, 
PhD, from Edanz Group (https://en-author-services.edanz.com/ac) for editing a draft of this manuscript.

\newcommand{\etalchar}[1]{$^{#1}$}

\appendix
\section{Approximation by the trapezoidal rule}

In this section, we describe the derivation of the approximation
when the trapezoidal rule is applied to a matrix function represented by a circular contour integral.
Moreover, we derive the upper bound \eqref{eq:erroroffMA} of the error in the approximation.

We first derive the approximation \eqref{eq:f_MA}.
Let $\Gamma$ be a circle of radius $r$ centered at $z_0$.
The integral can be represented as follows:
\begin{align}
  \frac{1}{2\pi\im}\int_\Gamma f(z) (zI - A)^{-1} \dz
  &= \frac{1}{2\pi\im} \int_{|z - z_0| = r} f(z)(zI - A)^{-1} \dz \notag \\
  &= \frac{1}{2\pi\im} \int_0^{2\pi} f(z_0 + r\e^{\im\theta}) ((z_0 + r\e^{\im\theta})I - A)^{-1} r\im\e^{\im\theta} \dtheta \notag \\
  &= \frac{1}{2\pi} \sum_{k=0}^{M-1} \int_{\theta_k}^{\theta_{k+1}} h(\theta) \dtheta \label{eq:true},
\end{align}
where $h(\theta) = f(z_0 + r\e^{\im\theta}) ((z_0 + r\e^{\im\theta})I - A)^{-1} r\e^{\im\theta} $
and $\theta_k = \frac{2\pi}{M}k \ (k=0,1,\dots,M-1)$.
We take the $M$ points at regular intervals and apply the trapezoidal rule $\int_a^b h(z) \dz \approx (b-a)\frac{h(b) - h(a)}{2}$.
Then, \Cref{eq:true} can be approximated as follows:
\begin{align}
  \frac{1}{2\pi} \sum_{k=0}^{M-1} \int_{\theta_k}^{\theta_{k+1}} h(\theta) \dtheta
  &\approx \frac{1}{M}\sum_{k=0}^{M-1} h(\theta_k) \notag \\
  &= \frac{1}{M} \sum_{k=0}^{M-1}   f(z_0 + r\e^{\im\theta_k}) ((z_0 + r\e^{\im\theta_k})I - A)^{-1} r\e^{\im\theta_k}.
\end{align}
Therefore, we obtain \Cref{eq:f_MA}.

Next, we show the upper bound of the error in the approximation.
\begin{proof}
  From $(z_0 + r\e^{\im\theta_k})I - A = r\e^{\im\theta_k}I - (-z_0 I +A )$ and  $r > \norm{-z_0 I + A}$,
  we can see that $((z_0 + r\e^{\im\theta_k})I - A)^{-1} = (r\e^{\im\theta_k})^{-1} \sum_{\ell = 0}^\infty ((-z_0 I + A)/(r\e^{\im\theta_k}))^\ell$.
  Let $f(z) = \sum_{n=0}^\infty a_n (z - z_0)^n $ be the Taylor expansion of $f(z)$.
  We have
  \begin{align}
    f_M(A)
    &= \frac{1}{M}\sum_{k=0}^{M-1} f(z_0 + r\e^{\im\theta_k}) \sum_{\ell=0}^\infty \left( \frac{-z_0 I + A}{r\e^{\im\theta_k}} \right)^\ell \notag \\
    &= \frac{1}{M}\sum_{k=0}^{M-1} \sum_{n=0}^\infty a_n (r\e^{\im\theta_k})^n \sum_{\ell=0}^\infty \left( \frac{-z_0 I + A}{r\e^{\im\theta_k}} \right)^\ell \notag \\
    &=  \sum_{n=0}^\infty\sum_{\ell=0}^\infty  a_n (-z_0 I + A)^\ell r^{n- \ell} \left(\frac{1}{M}\sum_{k=0}^{M-1} \e^{\im\theta_k(n - \ell)}\right).
  \end{align}
  We divide the sum into three parts for $n = \ell, n < \ell, $ and $n > \ell$:
  \begin{align}
    f_M(A) &=
    \sum_{n=0}^\infty a_n (-z_0 I + A)^n \label{first_term}  \\
    &\quad + \sum_{n   =0}^\infty \sum_{\ell=n+1}^\infty a_n (-z_0 I + A)^\ell r^{n - \ell} S_M(n - \ell)  \label{second_term} \\
    &\quad + \sum_{\ell=0}^\infty \sum_{n=\ell+1}^\infty a_n (-z_0 I + A)^\ell r^{n - \ell} S_M(n - \ell), \label{third_term}
  \end{align}
  where $ S_M(n-\ell) = \frac{1}{M}\sum_{k=0}^{M-1} \e^{\im\theta_k(n - \ell)}$.
  The first part \eqref{first_term} coincide with $f(A)$.  Therefore, we have
  \begin{align}
    f_M(A) = f(A) + \sum_{n   =0}^\infty \sum_{y = 1}^\infty a_n (-z_0 I + A)^{n+yM} r^{-yM}
    + \sum_{\ell=0}^\infty \sum_{y=1}^\infty a_{\ell + yM} (-z_0 I + A)^\ell r^{yM}.
  \end{align}
  The norm of the second part \eqref{second_term} is bounded as follows:
  \begin{align}
    \bigDnorm{ \sum_{n   =0}^\infty \sum_{y = 1}^\infty a_n (-z_0 I + A)^{n+yM} r^{-yM}  }
    &\le
    \sum_{n=0}^\infty \sum_{y=1}^\infty \abs{a_n} \norm{-z_0 I + A}^{n+yM} r^{-yM} \notag \\
    &\le
    \sum_{n=0}^\infty \sum_{y=1}^\infty \norm{f}_\infty \left(\frac{\norm{-z_0I + A}}{R}\right)^n \left(\frac{\norm{-z_0I + A}}{r}\right)^{yM} \notag \\
    &=
    \frac{\norm{f}_\infty}{1 - \frac{\norm{-z_0 I + A}}{R}}  \frac{1}{1 - \left(\frac{\norm{-z_0 I + A}}{r}\right)^M} \left(\frac{\norm{-z_0 I + A}}{r}\right)^M.
  \end{align}
  where we used that $\abs{a_n} \le \norm{f}_\infty/R^n$.
 The norm of the third part \eqref{third_term} is bounded as follows:
  \begin{align}
    \bigDnorm{
      \sum_{\ell=0}^\infty \sum_{y=1}^\infty a_{\ell + yM} (-z_0 I + A)^\ell r^{yM}
    }
    &\le
    \sum_{\ell=0}^\infty \sum_{y=1}^\infty \abs{a_{\ell + yM}} \norm{-z_0 I + A}^\ell r^{yM} \notag \\
    &\le
    \sum_{\ell=0}^\infty \sum_{y=1}^\infty \norm{f}_\infty \left(\frac{\norm{-z_0 I + A}}{R}\right)^\ell \left(\frac{r}{R}\right)^{yM} \notag \\
    &\le
    \frac{\norm{f}_\infty}{1 - \frac{\norm{-z_0 I + A}}{R}}  \frac{1}{1 - \left(\frac{r}{R}\right)^M} \left(\frac{r}{R}\right)^M.
  \end{align}
  Therefore, $\norm{f(A) - f_M(A)}$ can be bounded as \Cref{eq:erroroffMA}.
\end{proof}

\end{document}